\pdfoutput=1
\documentclass[journal]{IEEEtran}

\usepackage{amsmath,amsfonts}
\usepackage{amssymb}
\usepackage[ruled,vlined,linesnumbered]{algorithm2e}
\usepackage{array}
\usepackage{subcaption}
\usepackage{textcomp}
\usepackage{stfloats}
\usepackage{url}
\usepackage{verbatim}
\usepackage{graphicx}
\usepackage{cite}
\usepackage{multirow}
\usepackage{lipsum}
\usepackage{mathtools}
\usepackage{cuted}
\usepackage{breqn}
\usepackage[english]{babel}
\usepackage[autostyle]{csquotes}
\usepackage[acronym,toc,shortcuts]{glossaries}
\usepackage{comment}
\usepackage{siunitx}
\usepackage{textgreek}
\usepackage{xcolor}
\usepackage{float}

\newtheorem{lemma}{Lemma}
\newtheorem{theorem}{Theorem} 

\begin{document}

\title{Hierarchical Multi-Task Federated Learning in VANETs}

\author{M. Saeid HaghighiFard,~\IEEEmembership{Graduate Student Member,~IEEE} and Sinem Coleri,~\IEEEmembership{Fellow,~IEEE}
		
\thanks{M. Saeid Haghighifard and Sinem Coleri are with the Department of Electrical and Electronics Engineering, Koc University, Istanbul, Turkey, email: \{mhaghighifard21, scoleri\}@ku.edu.tr. This work is supported by the Scientific and Technological Research Council of Turkey, Grant number 119C058, and Ford Otosan.}
}

\markboth{Journal of \LaTeX\ Class Files,~Vol.~14, No.~8, August~2021}%
{Shell \MakeLowercase{\textit{et al.}}: A Sample Article Using IEEEtran.cls for IEEE Journals}

 \maketitle
\IEEEtitleabstractindextext{%
\begin{abstract}
Vehicular Ad hoc Networks (VANETs) increasingly rely on federated learning (FL) to enable collaborative intelligence without sharing raw sensory data. However, most existing vehicular FL frameworks assume that all vehicles train a single global model for a common task, which limits their applicability in practical vehicular environments where vehicles may perform heterogeneous learning tasks under non-independent and identically distributed (non-IID) data, intermittent connectivity, and high mobility. To address these challenges, this paper proposes an AutoEncoder-based Reliability-Optimized Hierarchical Multi-Task Federated Learning (AERO-HMTFL) framework for dynamic multi-hop clustered VANETs. The proposed framework introduces a tri-weighted clustering metric that jointly considers vehicular mobility, shared-model similarity, and task affinity to produce mobility-stable, semantically aligned clusters. Each vehicle employs a split-model architecture comprising a shared autoencoder-based representation module and multiple task-specific heads, with only the shared autoencoder parameters exchanged while the task heads remain local. To improve robustness, cluster heads perform reliability-aware aggregation based on historical validation performance and participation frequency, while the Evolved Packet Core (EPC) conducts global shared-autoencoder fusion across clusters. Extensive simulations demonstrate that, compared with the multi-task federated learning benchmarks, AERO-HMTFL achieves up to 13\% higher sustained EPC-level accuracy, exhibits more stable learning dynamics, and reduces EPC-level packet transmissions by approximately 87-97\%. Under short-range connectivity, it also requires approximately 13-29\% fewer communication rounds to converge.
\end{abstract}
}

\maketitle
\IEEEdisplaynontitleabstractindextext

\begin{IEEEkeywords}
Hierarchical federated learning, multi-task federated learning, clustering, vehicular networks.
\end{IEEEkeywords}

\section{Introduction}

\IEEEPARstart{V}{ehicular} Ad hoc Networks (VANETs) increasingly rely on machine learning (ML) to support safety- and mobility-critical services, including object detection, trajectory prediction, traffic-flow estimation, and driving-scene understanding. These services exploit the rich sensory information generated by onboard cameras, LiDAR, radar, and other vehicular sensing systems~\cite{zhang2024fl_its}. Conventional centralized learning requires vehicles to transmit their locally generated data to a central server, which raises substantial concerns regarding raw-data exposure, communication overhead, and scalability. Federated learning (FL) has emerged as a promising alternative in which vehicles train models locally and exchange model parameters or updates instead of raw data. By retaining raw observations at the vehicles, FL limits their direct exposure and supports distributed learning across mobile edge and vehicular environments~\cite{lim2020fl_survey,nguyen2021fl_iot,elbir2022fl_vehicular, elbir2022channel}. However, most conventional FL methods assume that all participating clients optimize a single shared model for a common learning objective. This assumption is restrictive in practical VANETs, where vehicles may perform heterogeneous but related sensing and learning tasks. These tasks may differ in their input modalities, label spaces, data distributions, and optimization objectives~\cite{zhang2024fl_its,wang2023flcav,chellapandi2024cav_survey}. Forcing heterogeneous objectives into a single global model can therefore be ineffective because incompatible output spaces, conflicting gradients, and task-dependent feature requirements may interfere during aggregation, degrading task-specific performance.

\subsection{Related Work}

Multi-task federated learning (MTFL) replaces the single-global-model assumption with coordinated learning among heterogeneous but related tasks~\cite{smith2017fmtl,marfoq2021fedem,li2021ditto,yi2024rhfedmtl}. Existing MTFL methods generally assign a separate model to each client or learning task and exploit relationships among these models during training. The seminal formulation in~\cite{smith2017fmtl} treats each client as a distinct but related task and jointly optimizes the client models while accounting for communication cost, stragglers, and fault tolerance. The mixture-based method in~\cite{marfoq2021fedem} learns several shared models and constructs a personalized model for each client by weighting them according to how well they represent the client’s local data. Ditto~\cite{li2021ditto} learns a global reference model together with a personalized model for each client. Ditto is applied independently to each learning task and therefore represents personalized federated learning without cross-task collaboration. RHFedMTL~\cite{yi2024rhfedmtl} considers multiple tasks in a terminal-based station-cloud hierarchy and controls the number of local and intermediate iterations according to the available resource budget. Although these methods support personalized or task-specific learning, they maintain a separate complete model for each task. Consequently, they either coordinate complete task models or train them independently, rather than jointly learning an explicit shared component through which vehicles performing different tasks can exchange transferable knowledge.

Recent MTFL methods address this limitation by dividing the learning model into common and task-dependent components and sharing only the components that are compatible across tasks. FedHCA$^2$~\cite{lu2024fedhca2} aggregates encoder parameters across heterogeneous clients and separately coordinates their task-dependent decoders to address model dissimilarity. FedBone~\cite{chen2024fedbone} places a general model at the server and retains task-specific models at the clients, while using gradient projection to reduce conflicts among heterogeneous tasks. M-Fed~\cite{cao2025mfed} shares encoder-based knowledge among clients performing different tasks and preserves task-dependent components locally. However, these methods assume direct client-server or fixed terminal-edge-cloud coordination. They do not support shared representation learning over a dynamically vehicular environment, where vehicle associations change with mobility and task availability, and model updates may arrive intermittently.

Hierarchical federated learning (HFL) alleviates the high communication overhead and central-server bottleneck of flat FL by introducing intermediate aggregation tiers. HiFlash~\cite{wu2023hiflash} lets edge servers synchronously aggregate updates from nearby clients and asynchronously forward the edge models to the cloud, while controlling stale updates and client-edge associations. ShapeFL~\cite{deng2024shapefl} selects the nodes that act as edge aggregators and assigns clients to them to reduce communication cost while maintaining diverse data at each aggregator. CPFedAvg~\cite{liu2025cpfedavg} removes the single cloud aggregation server and allows multiple upper-tier servers to exchange and mix their model parameters. Heal~\cite{xu2025heal} assigns clients to edge servers and transmits only selected model layers for aggregation. DaWa~\cite{sai2025dawa} uses reinforcement learning to select participating clients and determine their aggregation weights according to model accuracy and training cost. Despite employing different mechanisms to improve efficiency, conventional HFL methods generally rely on relatively stable client-edge associations and hierarchically train a single shared model for a common task. They are therefore not designed for dynamic vehicular networks, where mobility continuously changes connectivity and vehicle participation. The vehicular HFL frameworks in~\cite{11000096,haghighifard2026securehfl} enable mobility-aware multi-hop coordination by dynamically organizing vehicles into clusters, aggregating local updates at cluster heads, and forwarding cluster-level models to the edge cloud. However, they assume that all vehicles optimize the same model for a common learning objective and therefore do not support heterogeneous tasks, cross-task knowledge sharing, or task-aware vehicle association.

\subsection{Main Contributions}

To address these limitations, this paper develops an AutoEncoder-based Reliability-Optimized Hierarchical Multi-Task Federated Learning (AERO-HMTFL) framework for dynamic multi-hop clustered VANETs. The framework enables vehicles with heterogeneous learning objectives to exchange and hierarchically aggregate a shared autoencoder that captures transferable representations, while retaining raw data and task-specific components locally. By integrating task-aware clustering and reliability-weighted aggregation, it supports scalable and robust collaboration under mobility, intermittent connectivity, and time-varying participation.

The main contributions of this paper are summarized as follows.

\begin{itemize}
\item We propose a vehicular HFL that shifts from conventional single-task hierarchical learning to a multi-task architecture that supports heterogeneous learning objectives for vehicles in a privacy-preserving and scalable manner, for the first time in the literature.

\item We introduce, for the first time, task-related information into vehicle clustering for hierarchical multi-task learning in multi-hop VANETs. Specifically, the proposed tri-weighted metric jointly captures mobility and task-set dissimilarities, allowing vehicles to form clusters that are both mobility-stable and learning-compatible. 

\item We incorporate, for the first time in the literature, reliability-aware hierarchical aggregation into a task-aware MTFL framework for dynamic multi-hop VANETs. The mechanism gives greater influence to vehicles with consistently strong validation performance and regular participation, thereby limiting the influence of unstable, low-quality, or sporadic updates.

\item We introduce a split-model MTFL architecture tailored to hierarchical multi-hop VANETs. Shared autoencoder parameters are exchanged and hierarchically aggregated, while task-specific components remain local. This is the first framework to combine this architecture with task-aware clustering and reliability-weighted hierarchical aggregation in dynamic multi-hop VANETs.

\item We validate the effectiveness of the proposed framework through a comprehensive evaluation of task-aware and reliability-guided hierarchical MTFL in dynamic multi-hop VANETs. The results demonstrate consistent improvements in learning performance, convergence, robustness, and communication efficiency over representative MTFL baselines.

\end{itemize}

The remainder of this paper is organized as follows. Section~II presents the system model. Section~III describes the proposed hierarchical MTFL framework and its operation across vehicles, local aggregators, and the EPC. Section~IV provides the convergence analysis. Section~V presents the simulation setup and performance evaluation. Section~VI concludes the paper.

\section{System Model}

We consider a dynamic vehicular network comprising vehicles equipped with vehicle-to-vehicle (V2V), vehicle-to-infrastructure (V2I), and vehicle-to-network (V2N) communication interfaces. V2V communication may be supported by IEEE 802.11p~\cite{5888501}, IEEE 802.11bd~\cite{9779322}, or LTE-based Device-to-Device communication~\cite{7497762}, whereas V2I/V2N connectivity is provided through 5G New Radio V2X~\cite{9392787}.

Vehicles are organized into dynamic multi-hop clusters, each comprising one Cluster Head (CH) and a set of Cluster Members (CMs). Cluster membership and vehicle roles are updated based on network topology and mobility changes. The clustering procedure employs a tri-weighted association criterion based on relative mobility, similarity between shared encoder parameters, and task affinity derived from the active task sets of neighboring vehicles. This criterion promotes mobility-stable and task-coherent clusters.

Each vehicle owns a private non-IID dataset and supports a set of learning tasks. Its local model consists of a shared autoencoder, parameterized by encoder and decoder parameters, and a set of task-specific heads associated with learning tasks. The autoencoder learns a transferable latent representation across vehicles and tasks, while the task-specific heads remain local to preserve personalization and task-specific information.

The learning architecture follows a two-level hierarchy. CMs train their local autoencoder and task-specific heads and transmit only the updated autoencoder parameters and the required task and quality information to their CH. Each CH aggregates the received autoencoder updates to construct a cluster-level model. In every communication round, the CHs forward their cluster-level autoencoders to the Evolved Packet Core (EPC), which produces the global shared autoencoder and returns it to the CHs. Each CH then disseminates the received global autoencoder to its associated CMs, which use it to initialize the shared autoencoder in the next learning round. Raw data and task-specific head parameters are not exchanged.

Each vehicle maintains a Vehicle Information Base (VIB) containing its cluster role, hop count, mobility information, neighboring-vehicle information, active task set, and local model parameters. CHs additionally maintain participation and model-quality information for their associated CMs. The VIB is updated through vehicle state transitions and periodic ``HELLO\_PACKET'' exchanges, enabling distributed cluster formation and maintenance in the dynamic vehicular network.

\section{Multi-Task Federated Learning (MTFL) in Hierarchical Multi-Hop Clustered VANETs}

The proposed hierarchical MTFL framework comprises three interacting components: shared-representation learning at the vehicles, task-aware multi-hop cluster formation, and reliability-aware hierarchical aggregation. Vehicles are dynamically organized into multi-hop clusters according to a tri-weighted association metric that jointly considers relative mobility, shared-encoder similarity, and task-set affinity. Within each cluster, every participating CM jointly trains its shared autoencoder and local task-specific heads on its private non-IID data, while only the updated autoencoder parameters, the active task set, and the aggregate validation score are transmitted to the associated CH. The CH evaluates the reliability of the received updates based on historical validation performance and participation frequency, and constructs a reliability-weighted cluster-level autoencoder. In every communication round, the active CHs forward their cluster-level autoencoders to the EPC, where they are fused into a shared global autoencoder and then disseminated back through the hierarchy.

\subsection{Shared Autoencoder and Task-Specific Learning}

Each participating CM $i$ supports a nonempty set of learning tasks $\mathcal{T}_i$. Since different tasks may have different inputs and outputs, $x_{i,t}$ and $y_{i,t}$ denote, respectively, the input and target associated with task $t\in\mathcal{T}_i$. The task-dependent preprocessing used by the corresponding local learner maps each raw sample to the common input representation accepted by the shared autoencoder. The shared autoencoder contains an encoder and a decoder with parameters $\theta_i^{\mathrm{enc}}$ and $\theta_i^{\mathrm{dec}}$, respectively, while
$\{\phi_{i,t}\}_{t\in\mathcal{T}_i}$ denotes the local task-specific heads~\cite{cao2025mfed, collins2021fedrep, lu2024fedhca2}.

To make the round dependence explicit, let $r$ denote the communication round and $s\in\{0,\ldots,\tau-1\}$ the local SGD-step index. Define the shared-autoencoder parameter block and the local task-head block as
$\theta_i^{\mathrm{AE},(r,s)}=\{\theta_i^{\mathrm{enc},(r,s)},\theta_i^{\mathrm{dec},(r,s)}\}$ and
$\Phi_i^{(r,s)}=\{\phi_{i,t}^{(r,s)}\}_{t\in\mathcal{T}_i}$, respectively. For task $t$, the encoder, decoder, and local task head operate as
\begin{align}
 z_{i,t}^{(r,s)}
 &=
\mathrm{Enc}(x_{i,t};\theta_i^{\mathrm{enc},(r,s)}), \\
 \hat{x}_{i,t}^{(r,s)}
 &=
 \mathrm{Dec}(z_{i,t}^{(r,s)};\theta_i^{\mathrm{dec},(r,s)}), \\
 \hat{y}_{i,t}^{(r,s)}
 &=
 \mathrm{Head}_t(z_{i,t}^{(r,s)};\phi_{i,t}^{(r,s)}),
 \qquad t\in\mathcal{T}_i.
\end{align}

For task $t$, $\ell_t(\hat{y}_{i,t}^{(r,s)},y_{i,t})$ denotes the task-appropriate prediction loss, and
\begin{equation}
 L_{i,t}^{(r,s)}
 =
 \ell_t(\hat{y}_{i,t}^{(r,s)},y_{i,t}).
\end{equation}
Let $\xi_i^{(r,s)}$ denote the collection of mini-batches sampled by CM $i$ at local step $s$ in communication round $r$. The reconstruction loss is $\ell_{\mathrm{rec}}(\hat{x}_{i,t}^{(r,s)},x_{i,t})$, and the local multi-task objective is
\begin{equation}
\label{eq:localobjective}
\begin{split}
 &L_i^{(r,s)}\!\left(
 \theta_i^{\mathrm{AE},(r,s)},
 \Phi_i^{(r,s)};
 \xi_i^{(r,s)}
 \right)\\
 &\quad=
 \sum_{t\in\mathcal{T}_i}
 \left[
 L_{i,t}^{(r,s)}
 +
 \mu\,\ell_{\mathrm{rec}}(
 \hat{x}_{i,t}^{(r,s)},x_{i,t})
 \right],
\end{split}
\end{equation}
where $\mu\geq0$ controls the contribution of reconstruction learning. With learning rate $\eta>0$, and with $\nabla_{\vartheta}$ denoting the gradient with respect to parameter block $\vartheta$, the local SGD updates are
\begin{align}
 \theta_i^{\mathrm{enc},(r,s+1)}
 &\leftarrow
 \theta_i^{\mathrm{enc},(r,s)}
 -
 \eta\nabla_{\theta_i^{\mathrm{enc}}}L_i^{(r,s)},
 \label{eq:encoder_update}\\
 \theta_i^{\mathrm{dec},(r,s+1)}
 &\leftarrow
 \theta_i^{\mathrm{dec},(r,s)}
 -
 \eta\nabla_{\theta_i^{\mathrm{dec}}}L_i^{(r,s)},
 \label{eq:decoder_update}\\
 \phi_{i,t}^{(r,s+1)}
 &\leftarrow
 \phi_{i,t}^{(r,s)}
 -
 \eta\nabla_{\phi_{i,t}}L_{i,t}^{(r,s)},
 \quad t\in\mathcal{T}_i.
 \label{eq:head_update}
\end{align}

After completing the $\tau$ local SGD steps in the communication round $r$, CM $i$ sets
\begin{equation}
\theta_i^{\mathrm{AE},(r)}
=
\theta_i^{\mathrm{AE},(r,{\tau})}
=
\left\{
\theta_i^{\mathrm{enc},(r,{\tau})},
\theta_i^{\mathrm{dec},(r,{\tau})}
\right\}
\end{equation}
and evaluates each supported task head on the task-matched portion of the public validation buffer $\mathcal{D}_{\mathrm{pub}}$. Let $\mathrm{Acc}_{i,t}^{(r)}$ denote the resulting validation accuracy for task $t$. The aggregate validation score is
\begin{equation}
\label{eq:qualityscore}
 q_i^{(r)}
 =
 \frac{
 \sum_{t\in\mathcal{T}_i}
 \omega_{i,t}\,\mathrm{Acc}_{i,t}^{(r)}
 }{
 \sum_{t\in\mathcal{T}_i}\omega_{i,t}
 }.
\end{equation}
Here, $\omega_{i,t}=1$ gives equal task weighting, whereas $\omega_{i,t}=n_{i,t}^{\mathrm{val}}$ gives validation-sample weighting, where $n_{i,t}^{\mathrm{val}}$ is the number of public validation samples for task $t$. CM $i$ then transmits $\big(\theta_i^{\mathrm{AE},(r)},q_i^{(r)},\mathcal{T}_i\big)$ to its CH. The raw training data and the task-specific heads remain local.

For each CM, the CH maintains a participation-count-weighted historical validation average and the participation frequency:

\begin{align}
 \mathrm{HistAcc}_i^{(r)}
 &=
 \frac{
 (\kappa_i^{(r)}-1)\mathrm{HistAcc}_i^{(r-1)}
 +q_i^{(r)}
 }{\kappa_i^{(r)}},
 \label{eq:histacc}\\
 \mathrm{Freq}_i^{(r)}
 &=
 \frac{\kappa_i^{(r)}}{r},
 \label{eq:freq}
\end{align}
where $\kappa_i^{(r)}$ is the number of communication rounds in which CM $i$ has participated up to and including round $r$, with $\mathrm{HistAcc}_i^{(0)}=0$. Its reliability score is then given by
\begin{equation}
\label{eq:reliability}
 \mathrm{Rel}_i^{(r)}
 =
 \lambda_{\mathrm{acc}}\mathrm{HistAcc}_i^{(r)}
 +
 \lambda_{\mathrm{freq}}\mathrm{Freq}_i^{(r)},
 \qquad
 \lambda_{\mathrm{acc}}+\lambda_{\mathrm{freq}}=1.
\end{equation}

In communication round $r$, let $C$ denote the set of participating CMs in the considered cluster. The CH obtains the reliability-normalized weights and the cluster-level autoencoder as
\begin{equation}
\label{eq:clusteragg}
 w_i^{(r)}
 =
 \frac{\mathrm{Rel}_i^{(r)}}
 {\sum_{j\in C}\mathrm{Rel}_j^{(r)}},
 \qquad
 \overline{\theta}_{C}^{\mathrm{AE},(r)}
 =
 \sum_{i\in C}
 w_i^{(r)}
 \theta_i^{\mathrm{AE},(r)}.
\end{equation}

The CH then forwards the aggregated cluster-level autoencoder $\overline{\theta}_{C}^{\mathrm{AE},(r)}$ to the EPC. Let $\mathcal{C}_r$ denote the set of active clusters participating in communication round $r$. The EPC computes the global shared autoencoder as
\begin{equation}
\label{eq:epcaggregation}
\theta_{\mathrm{EPC}}^{\mathrm{AE},(r)}
=
\frac{1}{|\mathcal{C}_r|}
\sum_{C\in\mathcal{C}_r}
\overline{\theta}_{C}^{\mathrm{AE},(r)}.
\end{equation}
The EPC returns $\theta_{\mathrm{EPC}}^{\mathrm{AE},(r)}$ to the active CHs in the same communication round. Each CH then sets
\begin{equation}
\theta_{\mathrm{init},C}^{\mathrm{AE},(r+1)}
\leftarrow
\theta_{\mathrm{EPC}}^{\mathrm{AE},(r)}
\end{equation}
and disseminates this initialization to its associated CMs. Therefore, CM training, CH-level aggregation, and EPC-level aggregation are all indexed by the same communication-round index $r$, and the resulting global autoencoder initializes the shared autoencoder for communication round $r+1$. The task-specific heads $\{\phi_{i,t}\}_{t\in\mathcal{T}_i}$ are neither transmitted nor modified during CH-level or EPC-level aggregation and therefore remain local throughout the hierarchical learning procedure.

\subsection{Task-Aware Multi-Hop Clustering}

Cluster formation follows the distributed multi-hop clustering procedure introduced in~\cite{11000096}. A vehicle starts in the INITIAL state, constructs its Vehicle Information Base through periodic ``HELLO\_PACKET'' exchanges, and then enters the STATE ELECTION (SE) state. In SE, it first requests an association with a reachable CH and joins the accepting CH with the lowest association cost. If no direct CH is available, it requests association through eligible CMs whose current hop count is below the maximum hop limit and joins the lowest-cost accepted route. If neither a direct nor a multi-hop route is available, the vehicle becomes a CH only when its election cost is lower than that of the neighboring SE vehicles; otherwise, it remains in SE and repeats the decision after its VIB is updated. A CM returns to SE when its connection to its parent CM or CH is not maintained for the prescribed timeout period. A CH evaluates neighboring CHs and merges with a more suitable one by becoming a CM only if that neighboring CH has a lower association cost and accepts the connection request; otherwise, it remains a CH. A CH that remains disconnected from all CMs for the CH timeout also returns to SE.

Compared with the mobility- and model-similarity criterion in~\cite{11000096}, the association criterion used here additionally incorporates active-task similarity. For vehicles $i$ and $j$, the tri-weighted association cost is
\begin{equation}
\label{eq:avgcosim}
\begin{split}
 \mathrm{AvgCoSim}_{ij}&=\alpha\,\Delta v_{ij}\\
 &\quad+\beta\left(1-\cos\!\left(\theta_i^{\mathrm{enc}},\theta_j^{\mathrm{enc}}\right)\right)\\
 &\quad+\gamma\left(1-\mathrm{TaskSim}_{ij}\right),
\end{split}
\end{equation}
where:
\begin{itemize}
\item $\Delta v_{ij}$ is the average relative speed between vehicles $i$ and $j$, capturing their mobility dissimilarity.
\item $\cos(\theta^{\mathrm{enc}}_i, \theta^{\mathrm{enc}}_j)$ is the cosine similarity between the shared encoder parameters of vehicles $i$ and $j$, quantifying their alignment in feature representation learning. Therefore, $1-\cos(\theta^{\mathrm{enc}}_i, \theta^{\mathrm{enc}}_j)$ represents their shared-encoder dissimilarity.
\item $\mathrm{TaskSim}_{ij}$ is the task affinity score, computed as the Jaccard similarity between the sets of active tasks:
\begin{equation}
\mathrm{TaskSim}_{ij} = \frac{|\mathcal{T}_i \cap \mathcal{T}_j|}{|\mathcal{T}_i \cup \mathcal{T}_j|}
\end{equation}
where $\mathcal{T}_i$ and $\mathcal{T}_j$ are the sets of active tasks supported by vehicles $i$ and $j$, respectively. Accordingly, $1-\mathrm{TaskSim}_{ij}$ represents the task-set dissimilarity between the two vehicles.
\item $\alpha, \beta$, and $\gamma \in [0,1]$ are tunable weights with $\alpha + \beta + \gamma = 1$, enabling flexible balancing between the influence of mobility dissimilarity, shared-encoder dissimilarity, and task-set dissimilarity in cluster formation.
\end{itemize}
A lower value of $\mathrm{AvgCoSim}_{ij}$ indicates a more suitable association between vehicles $i$ and $j$. The balance between these weights can be adjusted according to application requirements: increasing $\alpha$ places greater emphasis on mobility stability by penalizing associations between vehicles with high relative speeds, while a higher $\gamma$ places greater emphasis on matching learning objectives by penalizing vehicles with dissimilar active task sets, which is beneficial in environments with high task heterogeneity. The shared-encoder dissimilarity term ensures that vehicles grouped together are aligned not only in their mobility behavior but also in the structure of their learned representations, thereby fostering efficient collaborative learning. The inherited clustering procedure is thereby adapted to form clusters that are mobility-stable, representation-aligned, and task-compatible for hierarchical MTFL.

\subsection{Hierarchical Multi-Task Federated Learning Procedure}

After the clusters are formed and maintained by the task-aware multi-hop clustering mechanism described above, the proposed hierarchical MTFL procedure is executed through Algorithms~1--3. Algorithm~1 is executed independently at each participating CM. At the beginning of communication round $r$, the CM receives the shared-autoencoder initialization $\theta_{\mathrm{init},CM}^{\mathrm{AE},(r)}$ from its associated CH and uses it to initialize its local shared autoencoder $\theta_{CM}^{\mathrm{AE},(r,0)}$ (Lines~1--2). Since the task-specific heads are not exchanged through the hierarchy, each local head $\phi_{CM,t}^{(r,0)}$ is initialized using its locally retained value from the preceding communication round (Line~3).

The CM then performs $\tau$ local SGD steps on its private non-IID data (Lines~4--12). At each local step $s$, it samples a mini-batch for every supported task $t\in\mathcal{T}_{CM}$ (Lines~4--5). For each supported task, the CM maps the task input $x_{CM,t}$ to the latent representation $z_{CM,t}$ using the shared encoder, reconstructs the input through the shared decoder, and applies the corresponding local task-specific head to obtain the task prediction $\hat{y}_{CM,t}$ (Lines~6--9). The task-specific prediction loss $L_{CM,t}$ is then evaluated using the loss function associated with task $t$ (Line~10). After processing all supported tasks, the CM constructs the complete local multi-task objective by summing the task-specific prediction losses and their corresponding reconstruction-loss terms, as given in~\eqref {eq:localobjective} (Line~11). The shared encoder and decoder are updated using the complete local objective, whereas each task-specific head is updated using only its corresponding task loss, according to~\eqref{eq:encoder_update}--\eqref{eq:head_update} (Line~12).

After completing the $\tau$ local updates, the CM sets its round-$r$ shared-autoencoder output to $\theta_{CM}^{\mathrm{AE},(r)}=\theta_{CM}^{\mathrm{AE},(r,{\tau})}$ and evaluates all supported task heads on the task-matched portions of the public validation buffer $\mathcal{D}_{\mathrm{pub}}$ (Line~13). The resulting task-specific validation accuracies $\mathrm{Acc}_{CM,t}^{(r)}$ are combined into the aggregate validation score $q_{CM}^{(r)}$ according to~\eqref{eq:qualityscore} (Line~14). Finally, the CM transmits only $\big(\theta_{CM}^{\mathrm{AE},(r)},q_{CM}^{(r)},\mathcal{T}_{CM}\big)$ to its associated CH (Line~15). The raw local data and the task-specific heads remain at the CM and are not included in the transmitted update.

Algorithm~2 is executed at each CH during every communication round. The CH first identifies the set $C$ of CMs that participate in round $r$ and collects their updated shared autoencoders, aggregate validation scores, and active task sets (Lines~1--3). For each participating CM, the CH updates the participation-count-weighted historical validation average using the newly received score $q_{CM}^{(r)}$ according to~\eqref{eq:histacc} (Lines~4--5). It also updates the participation frequency based on the number of rounds in which that CM has participated, relative to the current round index (Line~6). The CH then combines the historical validation score and participation frequency to compute the reliability score $\mathrm{Rel}_{CM}^{(r)}$ using~\eqref{eq:reliability} (Line~7).

After the reliability scores have been computed for all participating CMs, the CH normalizes them to obtain the aggregation weights $w_{CM}^{(r)}$ and performs reliability-weighted aggregation of the received shared-autoencoder parameters to construct the cluster-level autoencoder $\overline{\theta}_{C}^{\mathrm{AE},(r)}$ according to~\eqref{eq:clusteragg} (Line~8). The CH then sends $\overline{\theta}_{C}^{\mathrm{AE},(r)}$ to the EPC and receives the global shared autoencoder $\theta_{\mathrm{EPC}}^{\mathrm{AE},(r)}$ generated in the same communication round (Line~9). It assigns the received global autoencoder to $\theta_{\mathrm{init},C}^{\mathrm{AE},(r+1)}$ (Line~10) and broadcasts it to all associated CMs for initialization of communication round $r+1$ (Line~11). The task-specific heads are not involved in the CH-level or EPC-level aggregation and remain local.

Algorithm~3 is executed at the EPC during every communication round $r$. The EPC first collects the cluster-level autoencoders $\{\overline{\theta}_{C}^{\mathrm{AE},(r)}\}_{C\in\mathcal{C}_r}$ from the active CHs, where $\mathcal{C}_r$ denotes the set of active clusters in round $r$ (Lines~1--2). It then uniformly averages the received cluster-level autoencoders to construct the global shared autoencoder $\theta_{\mathrm{EPC}}^{\mathrm{AE},(r)}$ according to~\eqref{eq:epcaggregation} (Line~3). Finally, the EPC sends the resulting global autoencoder to every active CH in the same round, and each CH disseminates it to its associated CMs via Algorithm~2 to initialize round $r+1$ (Line~4). All task-specific heads remain local and unchanged under hierarchical aggregation.

\begin{algorithm}[ht]
\caption{Local Training and Validation at Cluster Member (CM) with Shared Autoencoder}
\For{each communication round $r$}{
    Receive $\theta_{\mathrm{init},CM}^{\mathrm{AE},(r)}$ from the associated CH and initialize $\theta_{CM}^{\mathrm{AE},(r,0)}\leftarrow\theta_{\mathrm{init},CM}^{\mathrm{AE},(r)}$\;
    Initialize $\phi_{CM,t}^{(r,0)}\leftarrow\phi_{CM,t}^{(r-1,{\tau})}$ for every supported task $t\in\mathcal{T}_{CM}$\;
    \For{$s=0,\ldots,\tau-1$}{
        Sample a mini-batch for every supported task $t\in\mathcal{T}_{CM}$\;
        \For{each local task $t\in\mathcal{T}_{CM}$}{
            $z_{CM,t}=\mathrm{Enc}(x_{CM,t};\theta_{CM}^{\mathrm{enc},(r,s)})$\;
            $\hat{x}_{CM,t}=\mathrm{Dec}(z_{CM,t};\theta_{CM}^{\mathrm{dec},(r,s)})$\;
            $\hat{y}_{CM,t}=\mathrm{Head}_t(z_{CM,t};\phi_{CM,t}^{(r,s)})$\;
            $L_{CM,t}=\ell_t(\hat{y}_{CM,t},y_{CM,t})$\;
        }
        $L_{CM}=\sum_{t\in\mathcal{T}_{CM}}\left[L_{CM,t}+\mu\ell_{\mathrm{rec}}(\hat{x}_{CM,t},x_{CM,t})\right]$\;
        Update $\theta_{CM}^{\mathrm{enc},(r,s+1)}$, $\theta_{CM}^{\mathrm{dec},(r,s+1)}$, and $\{\phi_{CM,t}^{(r,s+1)}\}_{t\in\mathcal{T}_{CM}}$ using~\eqref{eq:encoder_update}--\eqref{eq:head_update}\;
    }
    Set $\theta_{CM}^{\mathrm{AE},(r)}\leftarrow\theta_{CM}^{\mathrm{AE},(r,{\tau})}$ and evaluate all supported heads on the task-matched portions of $\mathcal{D}_{\mathrm{pub}}$\;
    Compute $\mathrm{Acc}_{CM,t}^{(r)}$ and $q_{CM}^{(r)}$ using~\eqref{eq:qualityscore}\;
    Send $\big(\theta_{CM}^{\mathrm{AE},(r)},q_{CM}^{(r)},\mathcal{T}_{CM}\big)$ to the CH\;
}
\end{algorithm}

\begin{algorithm}[ht]
\caption{Reliability-Weighted Aggregation at Cluster Head (CH)}
\For{each communication round $r$}{
    Let $C$ denote the set of participating CMs in the cluster during round $r$\;
    Collect $\big(\theta_{CM}^{\mathrm{AE},(r)},q_{CM}^{(r)},\mathcal{T}_{CM}\big)$ from every $CM\in C$\;
    \For{each participating $CM\in C$}{
        Update $\mathrm{HistAcc}_{CM}^{(r)}$ using~\eqref{eq:histacc}\;
        Update $\mathrm{Freq}_{CM}^{(r)}\leftarrow\kappa_{CM}^{(r)}/r$\;
        Compute $\mathrm{Rel}_{CM}^{(r)}$ using~\eqref{eq:reliability}\;
    }
    Compute $w_{CM}^{(r)}$ and $\overline{\theta}_{C}^{\mathrm{AE},(r)}$ using~\eqref{eq:clusteragg}\;
    Send $\overline{\theta}_{C}^{\mathrm{AE},(r)}$ to the EPC and receive $\theta_{\mathrm{EPC}}^{\mathrm{AE},(r)}$\;
    $\theta_{\mathrm{init},C}^{\mathrm{AE},(r+1)}\leftarrow\theta_{\mathrm{EPC}}^{\mathrm{AE},(r)}$\;
    Broadcast $\theta_{\mathrm{init},C}^{\mathrm{AE},(r+1)}$ to all associated CMs\;
}
\end{algorithm}

\begin{algorithm}[ht]
\caption{Global Autoencoder Aggregation at EPC}
\For{each communication round $r$}{
    Collect $\{\overline{\theta}_{C}^{\mathrm{AE},(r)}\}_{C\in\mathcal{C}_r}$ from all active CHs\;
    Compute $\theta_{\mathrm{EPC}}^{\mathrm{AE},(r)}=|\mathcal{C}_r|^{-1}\sum_{C\in\mathcal{C}_r}\overline{\theta}_{C}^{\mathrm{AE},(r)}$\;
    {Send $\theta_{\mathrm{EPC}}^{\mathrm{AE},(r)}$ to every active CH for dissemination to its associated CMs}\;
}
\end{algorithm}

\section{Convergence Analysis}

This section studies the convergence of the shared-autoencoder sequence
generated at the EPC. In Algorithms~1--3, the same index $r=1,2,\ldots$ denotes the communication round at the CM, CH, and EPC levels. Let $\theta^{(r)}$ denote the shared-autoencoder initialization used by the CMs at the beginning of round $r$, with $\theta^{(1)}$ denoting the initial shared autoencoder. The EPC output generated in round $r$ is denoted by $\theta^{(r+1)}\equiv\theta_{\mathrm{EPC}}^{\mathrm{AE},(r)}$ and initializes the CMs in round $r+1$. The active cluster partition is assumed to remain fixed during each communication round, although it may change between rounds. The analysis establishes an ergodic first-order stationarity guarantee for the EPC-level shared-autoencoder sequence.

\subsubsection{Global Objective and Relation to the Local Loss}

Let $\Phi_i=\{\phi_{i,t}\}_{t\in\mathcal{T}_i}$ denote the current local task-head block of CM $i$. In the generic objective and shared-gradient-oracle definitions below, the communication-round index $r$ and local-step index $s$ are omitted only for notational compactness; they are retained for all round- and step-dependent iterates, mini-batches, and aggregation weights. For the complete shared-autoencoder parameter block $\theta=\{\theta^{\mathrm{enc}},\theta^{\mathrm{dec}}\}$, define
\begin{equation}
F_i(\theta)
=
\mathbb{E}_{\xi_i}
\left[
L_i(\theta,\Phi_i;\xi_i)
\right],
\end{equation}
where $L_i(\theta,\Phi_i;\xi_i)$ is the parameterized local objective in~\eqref{eq:localobjective}. Only $\theta$ is the optimization variable in $F_i$; the effect of the concurrently updated local-head block $\Phi_i$ is incorporated through the shared-gradient oracle and the assumptions below. Throughout this section, $\mathbb{E}_{\xi_i}$ denotes expectation over a mini-batch $\xi_i\sim\mathcal{D}_i$ drawn from CM $i$'s local data distribution. An unsubscripted $\mathbb{E}$ denotes expectation over all algorithmic randomness, including mini-batch sampling, time-varying participation, cluster evolution, and adaptive aggregation weights, whereas $\mathbb{E}[\cdot\mid\mathcal{F}_r]$ denotes conditional expectation given the pre-round history $\mathcal{F}_r$ defined below. Define the mini-batch shared-gradient direction as
\begin{equation}
v_i(\theta;\xi_i)
\triangleq
\nabla_{\theta}L_i(\theta,\Phi_i;\xi_i),
\end{equation}
whereas $\nabla_{\theta}F_i(\theta)$ denotes the exact gradient of $F_i$ with respect to the shared-autoencoder parameter block $\theta$.

Let $N$ denote the total number of CMs in the vehicular population. For fixed target weights $p_i\geq0$ satisfying $\sum_{i=1}^{N}p_i=1$, define
\begin{equation}
F(\theta)
=
\sum_{i=1}^{N}p_iF_i(\theta).
\end{equation}
For any integer $R\geq1$, let $R$ denote the total number of communication rounds considered in the convergence analysis. The objective is to bound
\begin{equation}
\frac{1}{R}
\sum_{r=1}^{R}
\mathbb{E}
\left[
\left\|
\nabla_{\theta}F(\theta^{(r)})
\right\|^2
\right].
\end{equation}

\subsubsection{Equivalent EPC Update}

At the beginning of communication round $r$, participating CM $i$ initializes $\theta_i^{(r,0)}=\theta^{(r)}$ and performs $\tau$ local SGD steps:
\begin{equation}
\label{eq:conv_local_update}
\theta_i^{(r,s+1)}
=
\theta_i^{(r,s)}
-
\eta\,
v_i\!\left(
\theta_i^{(r,s)};
\xi_i^{(r,s)}
\right),
\quad
s=0,\ldots,\tau-1.
\end{equation}
Let $\mathcal{C}_r$ denote the set of active clusters in round $r$, and let $w_{i,C}^{(r)}$ denote the reliability-normalized CH weight of CM $i$ in cluster $C$. Since all participating CMs are initialized by $\theta^{(r)}$, the cluster-level autoencoder is
\begin{align}
\overline{\theta}_{C}^{(r)}
&=
\sum_{i\in C}
w_{i,C}^{(r)}
\theta_i^{(r,{\tau})}
\nonumber\\
&=
\theta^{(r)}
-
\eta
\sum_{i\in C}
w_{i,C}^{(r)}
\sum_{s=0}^{\tau-1}
v_i\!\left(
\theta_i^{(r,s)};
\xi_i^{(r,s)}
\right).
\end{align}
The EPC aggregates the cluster-level autoencoders in the same communication round:
\begin{equation}
\theta^{(r+1)}
=
\frac{1}{|\mathcal{C}_r|}
\sum_{C\in\mathcal{C}_r}
\overline{\theta}_{C}^{(r)}.
\end{equation}
For CM $i$ belonging to cluster $C$, define its effective EPC weight as
\begin{equation}
a_i^{(r)}
=
\frac{w_{i,C}^{(r)}}{|\mathcal{C}_r|},
\end{equation}
and set $a_i^{(r)}=0$ for a nonparticipating CM. Then $a_i^{(r)}\geq0$ and $\sum_{i=1}^{N}a_i^{(r)}=1$. Substitution gives the equivalent EPC update
\begin{equation}
\label{eq:epc_equivalent_update}
\theta^{(r+1)}
=
\theta^{(r)}
-
\eta H^{(r)},
\end{equation}
where
\begin{equation}
\label{eq:aggregated_direction}
H^{(r)}
=
\sum_{i=1}^{N}
a_i^{(r)}
\sum_{s=0}^{\tau-1}
v_i\!\left(
\theta_i^{(r,s)};
\xi_i^{(r,s)}
\right).
\end{equation}

\subsubsection{Assumptions}

\textbf{Assumption 1 (Smoothness):}
Each $F_i$ is differentiable and $L$-smooth:
\begin{equation}
\left\|
\nabla_{\theta}F_i(\theta)
-
\nabla_{\theta}F_i(\theta')
\right\|
\leq
L\left\|\theta-\theta'\right\|.
\end{equation}

\textbf{Assumption 2 (Stochastic Shared-Gradient Regularity):}
For some $G<\infty$,
\begin{equation}
\begin{aligned}
{\mathbb{E}_{\xi_i}}
\left[
v_i(\theta;\xi_i)
\mid\theta
\right]
&=
{\nabla_{\theta}}F_i(\theta),\\
{\mathbb{E}_{\xi_i}}
\left[
\left\|v_i(\theta;\xi_i)\right\|^2
\mid\theta
\right]
&\leq
G^2.
\end{aligned}
\end{equation}
Both conditional expectations are taken over the mini-batch draw $\xi_i\sim\mathcal{D}_i$ with a fixed shared autoencoder parameter $\theta$. Consequently, $\|{\nabla_{\theta}}F_i(\theta)\|\leq G$.

\textbf{Assumption 3 (Adaptive Aggregate-Noise Regularity):}
Define
\begin{equation}
\widetilde{H}^{(r)}
=
\sum_{i=1}^{N}
a_i^{(r)}
\sum_{s=0}^{\tau-1}
{\nabla_{\theta}}F_i\!\left(
\theta_i^{(r,s)}
\right).
\end{equation}
Let $\mathcal{F}_r$ denote the sigma-algebra generated by the initial shared autoencoder and by all cluster states, participation decisions, local task-head states, mini-batch samples, model iterates, validation scores, and reliability quantities from rounds $1,\ldots,r-1$, together with the active-cluster and participating-CM sets selected at the beginning of round $r$. The reliability-weighted stochastic error satisfies
\begin{equation}
\mathbb{E}
\left[
H^{(r)}-\widetilde{H}^{(r)}
\mid
\mathcal{F}_r
\right]
=
0
\end{equation}
and, for some $\sigma_{\mathrm{agg}}^2<\infty$,
\begin{equation}
\mathbb{E}
\left[
\left\|
H^{(r)}
-
\mathbb{E}
\left[
H^{(r)}
\mid
\mathcal{F}_r
\right]
\right\|^2
\middle|
\mathcal{F}_r
\right]
\leq
\tau\sigma_{\mathrm{agg}}^2.
\end{equation}
This assumption allows the adaptive reliability weights to depend on the current local updates.

\textbf{Assumption 4 (Reliability-Weight Regularity):}
Define the conditional mean effective weight as
\begin{equation}
\overline{a}_i^{(r)}
=
\mathbb{E}
\left[
a_i^{(r)}
\mid
\mathcal{F}_r
\right]
\end{equation}
and the mismatch between the effective reliability weights and the target objective weights as
\begin{equation}
\delta_r
=
\sum_{i=1}^{N}
\left|
\overline{a}_i^{(r)}-p_i
\right|.
\end{equation}

\textbf{Assumption 5 (Lower-Bounded Objective):}
There exists $F_\star>-\infty$ such that
\begin{equation}
F(\theta)\geq F_\star,
\qquad
\forall\theta.
\end{equation}

\subsubsection{Auxiliary Bounds}

\begin{lemma}[Local Drift]
Under Assumption~2, for every participating CM $i$, round $r$, and $s\in\{0,\ldots,\tau\}$, where $s=\tau$ denotes the post-update model produced by the final local update at $s=\tau-1$,
\begin{equation}
\label{eq:local_drift}
\mathbb{E}
\left[
\left\|
\theta_i^{(r,s)}-\theta^{(r)}
\right\|^2
\right]
\leq
\eta^2s^2G^2
\leq
\eta^2\tau^2G^2.
\end{equation}
\end{lemma}

\begin{IEEEproof}
Summing the first $s$ local updates in~\eqref{eq:conv_local_update} gives
\begin{equation}
\theta_i^{(r,s)}-\theta^{(r)}
=
-\eta
\sum_{u=0}^{s-1}
v_i\!\left(
\theta_i^{(r,u)};
\xi_i^{(r,u)}
\right).
\end{equation}
Minkowski's inequality and Assumption~2 yield
\begin{equation}
\left(
\mathbb{E}
\left[
\left\|
\theta_i^{(r,s)}-\theta^{(r)}
\right\|^2
\right]
\right)^{1/2}
\leq
\eta sG.
\end{equation}
Squaring both sides and using $s\leq\tau$ proves the result.
\end{IEEEproof}

Define the conditional mean direction and its error as
\begin{equation}
m^{(r)}
=
\frac{1}{\tau}
\mathbb{E}
\left[
H^{(r)}
\mid
\mathcal{F}_r
\right],
\qquad
e^{(r)}
=
m^{(r)}-{\nabla_{\theta}}F(\theta^{(r)}).
\end{equation}

\begin{lemma}[Gradient-Direction Error]
Under Assumptions~1--4,
\begin{equation}
\label{eq:direction_error}
\mathbb{E}
\left[
\left\|
e^{(r)}
\right\|^2
\right]
\leq
2L^2\eta^2\tau^2G^2
+
2G^2\delta_r^2.
\end{equation}
\end{lemma}

\begin{IEEEproof}
Add and subtract $\sum_i\overline{a}_i^{(r)}{\nabla_{\theta}}F_i(\theta^{(r)})$. By Assumption~1 and~\eqref{eq:local_drift}, the local-drift component is bounded in root mean square by $L\eta\tau G$. By Assumption~2 and the definition of $\delta_r$, the reliability-weight mismatch component is bounded by $G\delta_r$. Applying $\|u+v\|^2\leq2\|u\|^2+2\|v\|^2$ proves~\eqref{eq:direction_error}.
\end{IEEEproof}

\subsubsection{Main Convergence Result}

\begin{theorem}[Shared-Autoencoder Stationarity Bound]
Suppose Assumptions~1--5 hold and
\begin{equation}
\eta\tau
\leq
\frac{1}{4L}.
\end{equation}
Then the EPC iterates generated by~\eqref{eq:epc_equivalent_update} satisfy
\begin{multline}
\label{eq:main_conv_bound}
\frac{1}{R}
\sum_{r=1}^{R}
\mathbb{E}\!\left[
\left\|
{\nabla_{\theta}}F\!\left(\theta^{(r)}\right)
\right\|^2
\right]
\\
\leq
\frac{4\left(F\!\left(\theta^{(1)}\right)-F_\star\right)}
{\eta\tau R}
+
6L^2\eta^2\tau^2G^2
\\
\quad+
\frac{6G^2}{R}
\sum_{r=1}^{R}\delta_r^2
+
2L\eta\sigma_{\mathrm{agg}}^2.
\end{multline}
\end{theorem}

\begin{IEEEproof}
By the smoothness of $F$ and~\eqref{eq:epc_equivalent_update},
\begin{align}
F(\theta^{(r+1)})
&\leq
F(\theta^{(r)})
-
\eta
\left\langle
{\nabla_{\theta}}F(\theta^{(r)}),
H^{(r)}
\right\rangle
\nonumber\\
&\quad+
\frac{L\eta^2}{2}
\left\|H^{(r)}\right\|^2.
\end{align}
Let $h^{(r)}={\nabla_{\theta}}F(\theta^{(r)})$. Conditioning on $\mathcal{F}_r$ gives
\begin{equation}
\mathbb{E}
\left[
H^{(r)}
\mid
\mathcal{F}_r
\right]
=
\tau\left(h^{(r)}+e^{(r)}\right).
\end{equation}
Assumption~3 and the conditional variance decomposition give
\begin{align}
\mathbb{E}
\left[
\left\|H^{(r)}\right\|^2
\middle|
\mathcal{F}_r
\right]
&\leq
2\tau^2\left\|h^{(r)}\right\|^2
+
2\tau^2\left\|e^{(r)}\right\|^2
\nonumber\\
&\quad+
\tau\sigma_{\mathrm{agg}}^2.
\end{align}
Using $\langle h,h+e\rangle\geq\frac{1}{2}\|h\|^2-\frac{1}{2}\|e\|^2$ and $\eta\tau\leq(4L)^{-1}$ yields
\begin{align}
\mathbb{E}
\left[
F(\theta^{(r+1)})
\right]
&\leq
\mathbb{E}
\left[
F(\theta^{(r)})
\right]
-
\frac{\eta\tau}{4}
\mathbb{E}
\left[
\left\|{\nabla_{\theta}}F(\theta^{(r)})\right\|^2
\right]
\nonumber\\
&\quad+
\frac{3\eta\tau}{4}
\mathbb{E}
\left[
\left\|e^{(r)}\right\|^2
\right]
+
\frac{L\eta^2\tau}{2}
\sigma_{\mathrm{agg}}^2.
\end{align}
Summing over $r=1,\ldots,R$, applying Assumption~5, and substituting~\eqref{eq:direction_error} prove~\eqref{eq:main_conv_bound}.
\end{IEEEproof}

\textit{Convergence Remark:}
Theorem~1 shows that, with a constant learning rate, the EPC-level shared autoencoder approaches a stationary neighborhood of $F$. The term $6L^2\eta^2\tau^2G^2$ represents the local-model drift caused by the $\tau$ CM updates performed before CH and EPC aggregation in each communication round. The term $6G^2R^{-1}\sum_{r=1}^{R}\delta_r^2$ represents the mismatch between the adaptive reliability weights and the target objective weights, whereas $2L\eta\sigma_{\mathrm{agg}}^2$ represents aggregate stochastic variation.

The exact vanishing of the average squared gradient norm can be established by considering a sequence of constant learning rates $\{\eta_R\}_{R\geq1}$, where $\eta_R$ is the learning rate selected for an execution analyzed over $R$ communication rounds. For fixed $\tau$, if $\eta=\eta_R\rightarrow0$, $\eta_RR\rightarrow\infty$, and
\begin{equation}
\frac{1}{R}
\sum_{r=1}^{R}
\delta_r^2
\rightarrow0,
\end{equation}
then
\begin{equation}
\lim_{R\rightarrow\infty}
\frac{1}{R}
\sum_{r=1}^{R}
\mathbb{E}
\left[
\left\|
{\nabla_{\theta}}F(\theta^{(r)})
\right\|^2
\right]
=
0.
\end{equation}
Therefore, as the number of communication rounds grows, the average expected stationarity measure of the EPC-level shared autoencoder vanishes, implying convergence to first-order stationary points of $F$ in the ergodic sense.

\section{Performance Evaluation}

The simulations evaluate the performance of the proposed AutoEncoder-based Reliability-Optimized Hierarchical Multi-Task Federated Learning (AERO-HMTFL) algorithm against three representative benchmarks. RHFedMTL~\cite{yi2024rhfedmtl} performs hierarchical multi-task learning without a shared autoencoder by assigning each task to a base station. Vehicles within the corresponding coverage area transmit their task-specific model updates to that base station, which performs task-level aggregation, while the central server coordinates training across the hierarchy. For the personalized FL benchmark, the optimization principle of Ditto~\cite{li2021ditto} is adapted to the underlying structure of the proposed framework. Independent Ditto processes are executed in parallel for the different tasks, with each process learning a task-specific reference model and personalized local models for the vehicles performing that task. The third benchmark is M-Fed~\cite{cao2025mfed}, an autoencoder-based multi-task federated learning approach that employs an encoder–decoder architecture to facilitate deep feature extraction but does not utilize clustering. All clients participate in a single global aggregation stage without CH-level coordination or mobility-aware grouping. Together, these baselines enable a systematic evaluation of the individual contributions of autoencoder-based feature learning, hierarchical and mobility-aware clustering, and reliability-weighted multi-tier aggregation.

\subsection{Simulation Setup}

The simulations are implemented in Python, with vehicular mobility traces generated using the Simulation of Urban Mobility (SUMO)~\cite{SUMO}  and real-time data streaming handled through KAFKA~\cite{KAFKA}. This integrated setup provides a realistic and dynamically evolving vehicular network environment for federated learning experiments. SUMO accurately models individual driver behaviors, traffic flows, and mobility patterns, while KAFKA streams data packets and model parameters to emulate communication exchanges among vehicles in a distributed network. All federated learning models, including the proposed AERO-HMTFL and the benchmark baselines, are developed using PyTorch, ensuring efficient training, reproducibility, and seamless integration with mobility and communication components~\cite{haghighifard2026aerohmtfl}.


\subsection{Simulation Environment}

For communication modeling, IEEE~802.11p is used for vehicle-to-vehicle (V2V) communication, while 5G~NR links support vehicle-to-infrastructure (V2I) communication. V2V channels follow the Winner+~B1 propagation model, as described in~\cite{9599363}, ensuring realistic characterization of short-range vehicular communication. For V2I communication, the Friis propagation model is used, following~\cite{3gpptr38901}, which provides accurate path-loss estimation between vehicles and the 5G~NR base station.

To evaluate learning performance across multiple perception and classification tasks, we employ three widely used datasets, CIFAR-10, GTSRB, and MNIST, each representing a distinct level of complexity and non-IID data characteristics. CIFAR-10 consists of 60,000 RGB images across 10 object categories and serves as a benchmark for general-purpose vision tasks under moderate visual variability. GTSRB contains over 50,000 traffic sign images spanning 43 classes, capturing real-world conditions with variations in illumination, occlusion, and viewpoint, and is well aligned with vehicular perception scenarios. MNIST provides 70,000 handwritten digit images (28$\times$28) representing digits 0--9 and is used as a lightweight baseline to assess classification performance under low computational cost.

For the CIFAR-10 and GTSRB tasks, a lightweight convolutional neural network (LightCNN) is employed. The model comprises two convolutional layers (3$\rightarrow$16 and 16$\rightarrow$32), batch normalization, ReLU activations, max pooling, global average pooling, dropout, and a fully connected output layer. Training is performed using stochastic gradient descent (SGD) with a learning rate of 0.01 (decayed by 0.95 every 10 communication rounds), momentum of 0.9, and batch sizes of 64 and 128 for training and testing, respectively.

For the MNIST task, a linear \texttt{SGDClassifier} is used, trained with an adaptive learning-rate schedule and incremental updates via the \texttt{partial\_fit} mechanism for up to 100 iterations. This configuration allows efficient adaptation to client-specific non-IID data streams.

Additionally, for the model compression task, a linear autoencoder is employed. The autoencoder consists of encoder and decoder matrices, $W_{\text{enc}}$ and $W_{\text{dec}}$, with a latent dimension not exceeding 256. Training is conducted using SGD with a learning rate of $10^{-3}$, minimizing mean squared error (MSE) to learn compact latent representations that support shared feature extraction across heterogeneous tasks.

To reproduce the task heterogeneity, each vehicle $i$ is randomly
assigned a nonempty subset
\begin{equation}
\mathcal{T}_i
\subseteq
\left\{
\mathrm{MNIST},
\mathrm{CIFAR\mbox{-}10},
\mathrm{GTSRB}
\right\},
\qquad
1\leq |\mathcal{T}_i|\leq3.
\end{equation}
Accordingly, some vehicles support only one learning task, some support two, and the remaining support all three. For each vehicle, the number of supported tasks is first determined by the adopted random seed, after which the corresponding tasks are sampled without replacement from the three available tasks. This random subset assignment creates task heterogeneity across the vehicular population while still allowing vehicles with partially overlapping task sets to exchange transferable representations through the shared autoencoder.

The EPC aggregates only the shared-autoencoder parameters and does not store or aggregate task-specific heads. Therefore, the ``EPC accuracy'' reported for task $t$ in Figs.~1--4 is an offline task-wise evaluation metric associated with the EPC-generated shared autoencoder, rather than the output of a task classifier located at the EPC. Let $\mathcal{V}_t^{(r)}$ denote the set of participating CMs that support task $t$ in round $r$. After the global autoencoder is disseminated, each CM $i\in\mathcal{V}_t^{(r)}$ combines it with its locally retained task head and evaluates its task-matched test data. The reported value is
\begin{equation}
{
\mathrm{Acc}_{\mathrm{EPC},t}^{(r)}
=
\frac{1}{|\mathcal{V}_t^{(r)}|}
\sum_{i\in\mathcal{V}_t^{(r)}}
\mathrm{Acc}_{i,t}
\!\left(
\theta_{\mathrm{EPC}}^{\mathrm{AE},(r)},
\phi_{i,t}^{(r,\tau)};
\mathcal{D}_{i,t}^{\mathrm{test}}
\right).
}
\label{eq:epc_task_accuracy}
\end{equation}
This evaluation is performed by the simulation evaluator and does not require transmitting the local task heads to the EPC or introducing additional protocol packets.

The performance of AERO-HMTFL and all benchmark methods is evaluated using two primary metrics: model accuracy, which measures predictive performance~\cite{Goodfellow2016}, and convergence behavior, defined as the number of communication rounds required for stabilization. Convergence is declared when the improvement in model accuracy falls below a threshold $\epsilon$ for 5 consecutive rounds~\cite{mcmahan2017fedavg}.

\subsection{Performance Comparison of Proposed Algorithm with Benchmark Algorithms}
\begin{figure}[!t]
     \centering
     \begin{subfigure}{0.45\textwidth}
         \centering
         \includegraphics[width=\textwidth]{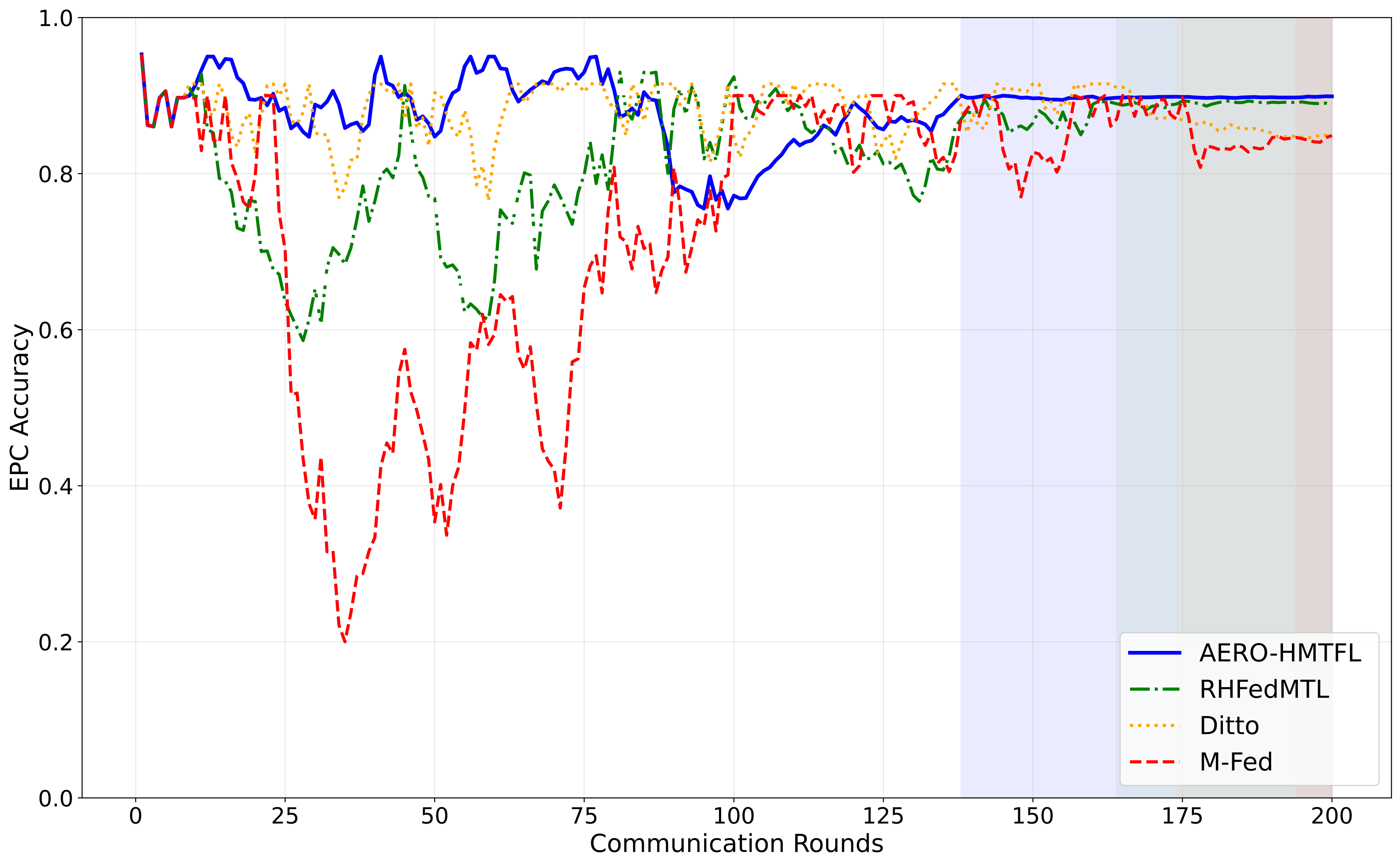}
         \subcaption*{(a)}
         \label{fig: MNIST_20}
     \end{subfigure}
     \vfill
     \begin{subfigure}{0.45\textwidth}
         \centering
         \includegraphics[width=\textwidth]{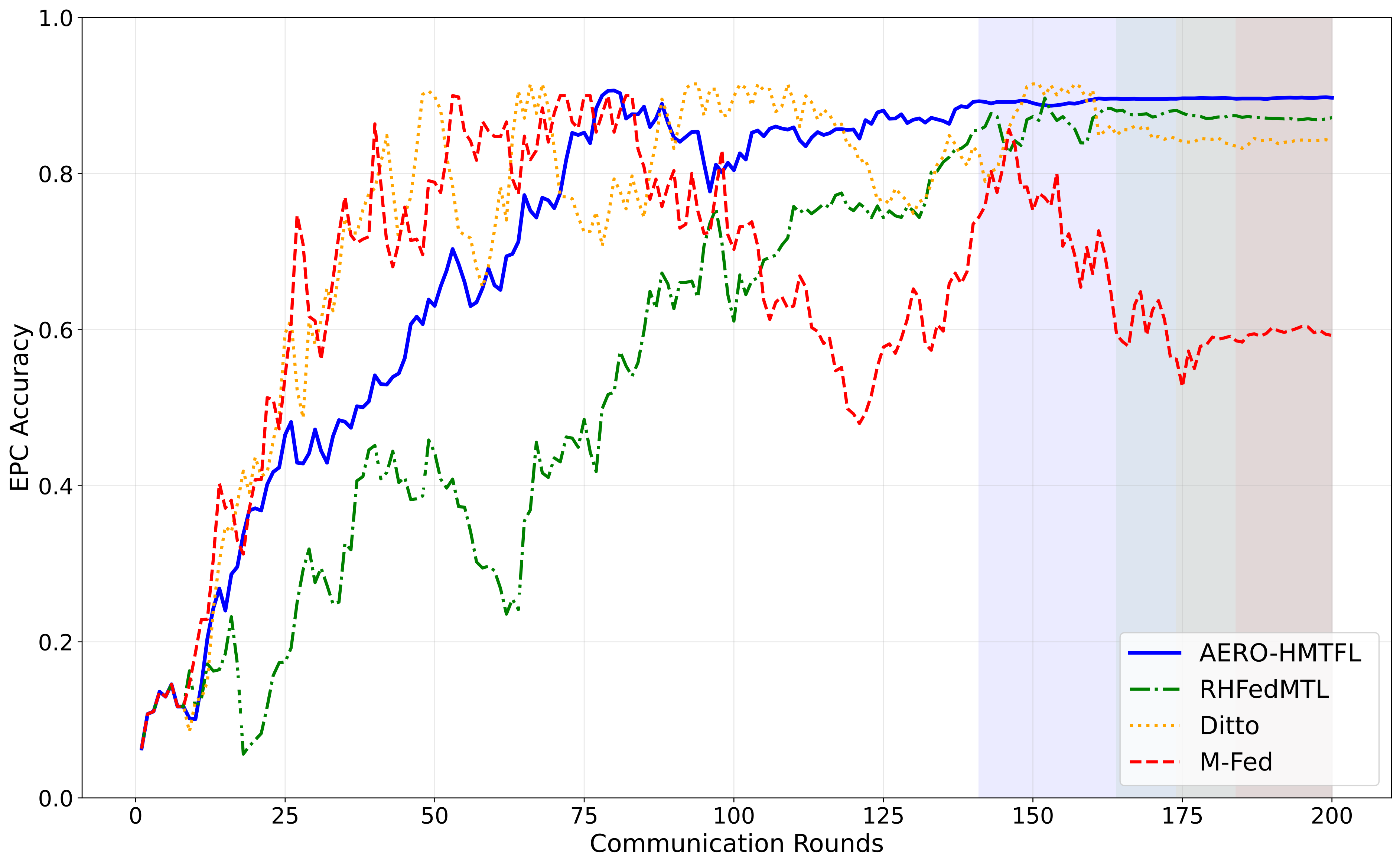}
         \subcaption*{(b)}
         \label{fig: GTSRB_20}
     \end{subfigure}
     \vfill
     \begin{subfigure}{0.45\textwidth}
         \centering
         \includegraphics[width=\textwidth]{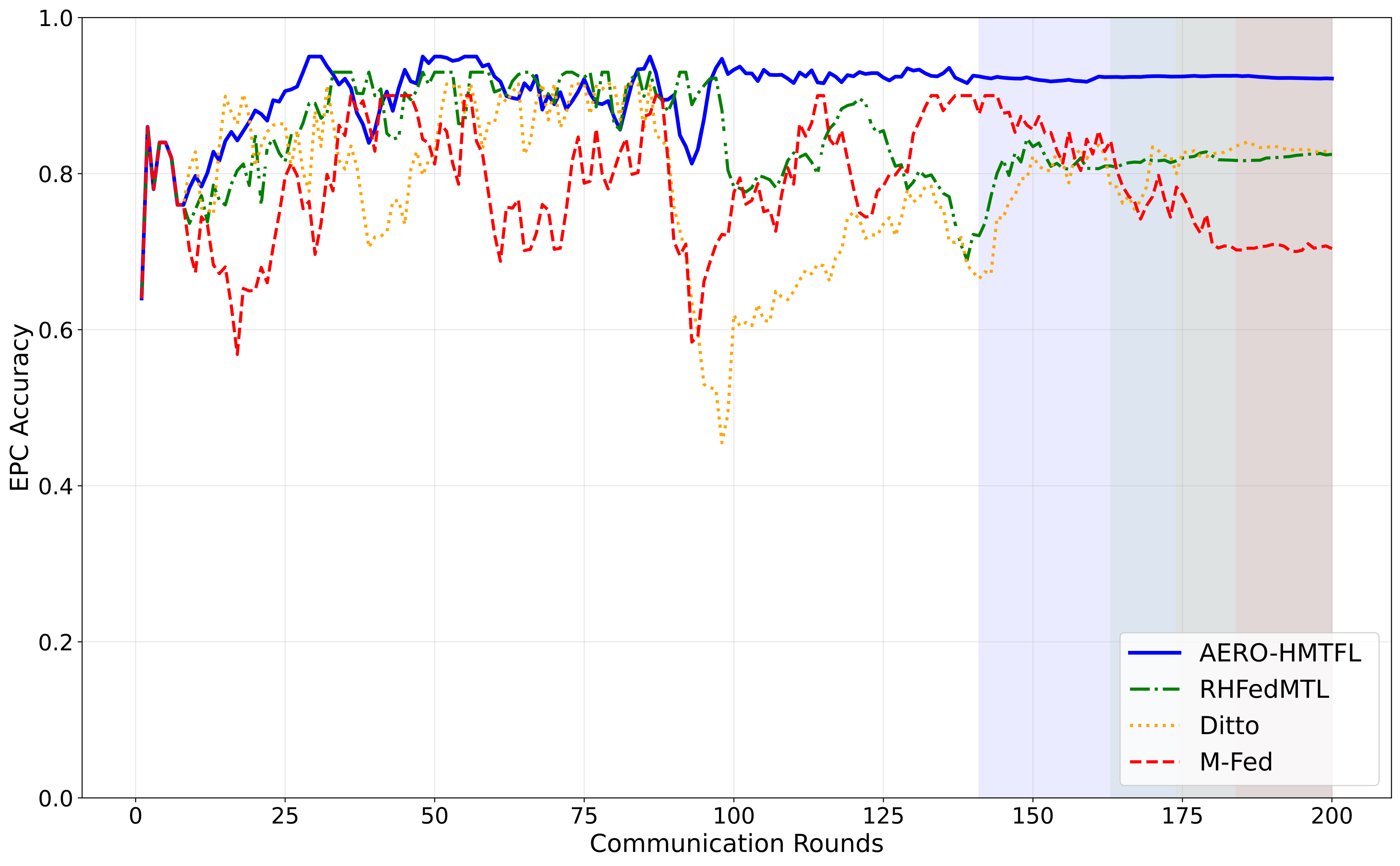}
         \subcaption*{(c)}
         \label{fig: CIFAR_20}
     \end{subfigure}
        \caption{Comparison of EPC accuracy across different tasks in the 20-vehicle scenario, with a 100 m transmission range and single-hop communication links. a) MNIST, b) GTSRB, c) CIFAR-10.}
        \label{fig: EPC Accuracy for 20 Vehicles}
\end{figure}
\begin{figure}[!t]
     \centering
     \begin{subfigure}{0.45\textwidth}
         \centering
         \includegraphics[width=\textwidth]{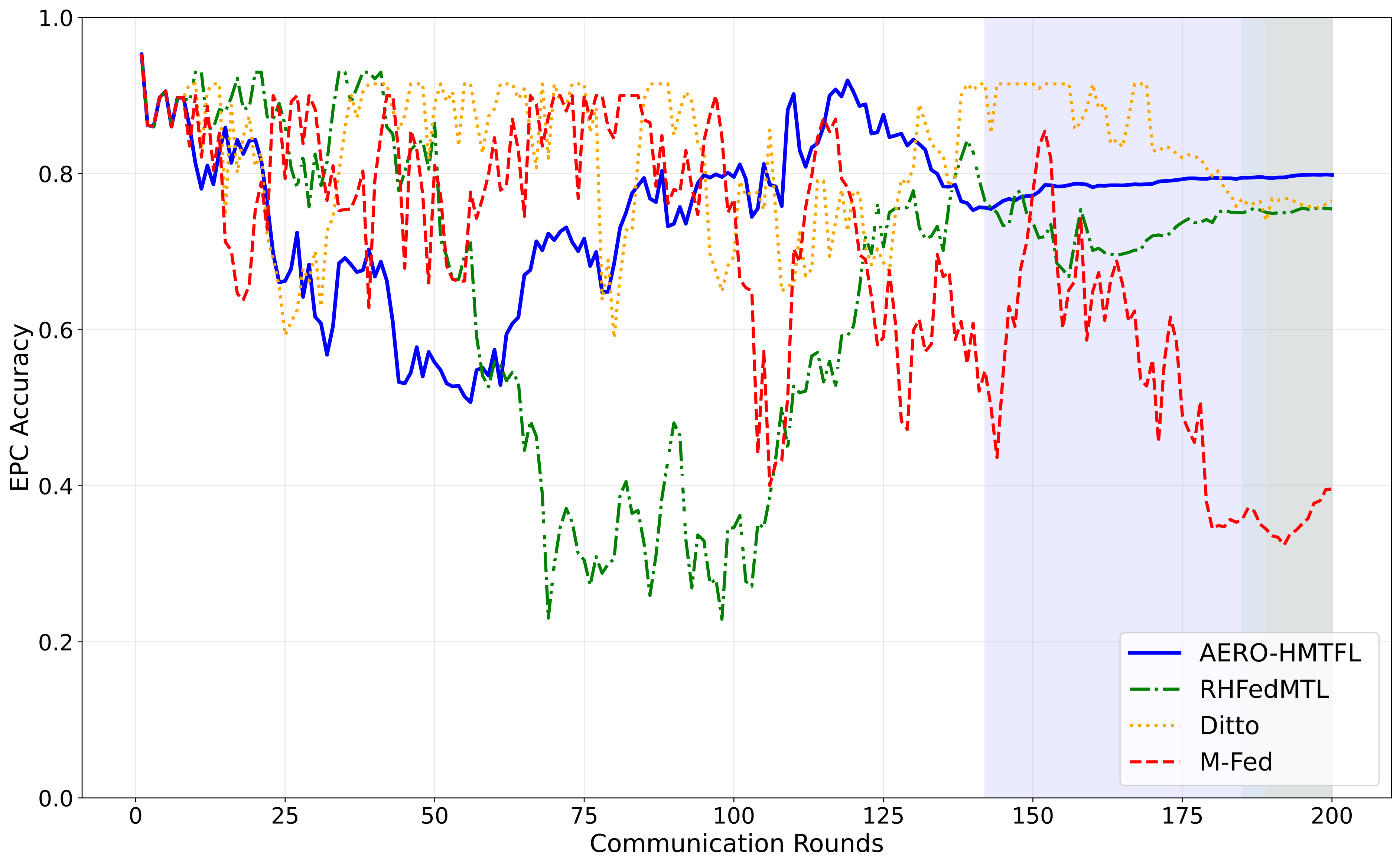}
         \subcaption*{(a)}
         \label{fig: MNIST_50}
     \end{subfigure}
     \vfill
     \begin{subfigure}{0.45\textwidth}
         \centering
         \includegraphics[width=\textwidth]{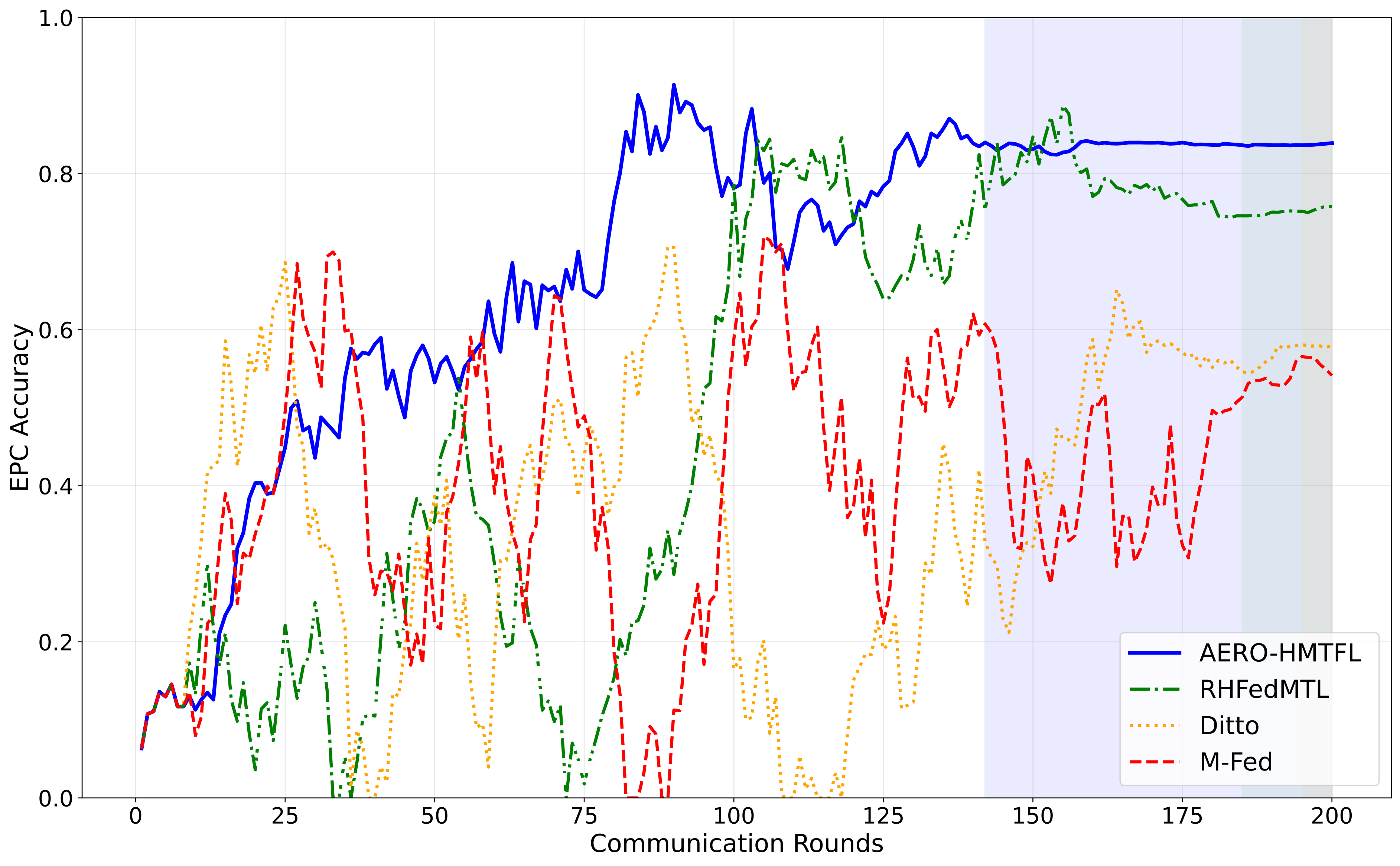}
         \subcaption*{(b)}
         \label{fig: GTSRB_50}
     \end{subfigure}
     \vfill
     \begin{subfigure}{0.45\textwidth}
         \centering
         \includegraphics[width=\textwidth]{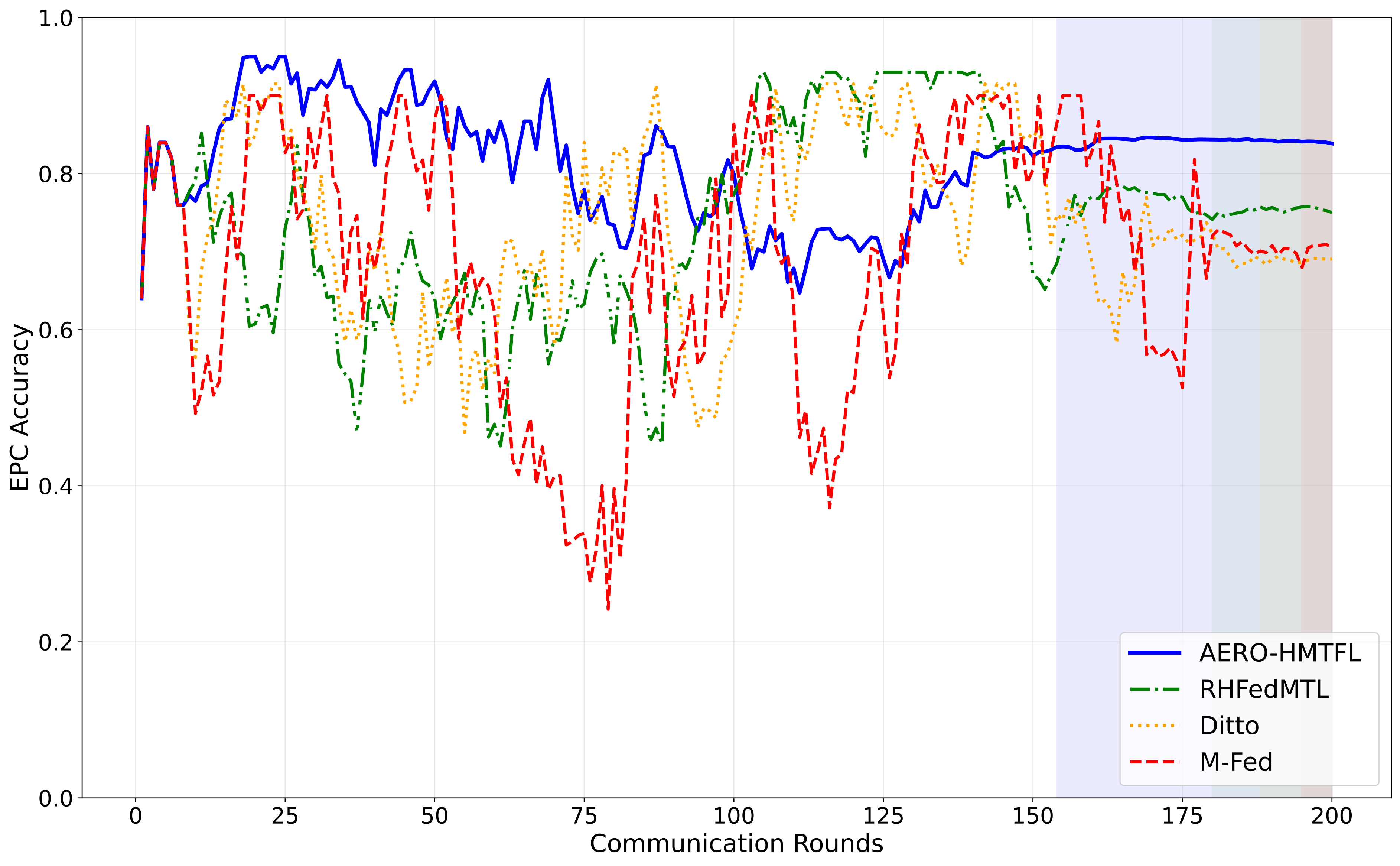}
         \subcaption*{(c)}
         \label{fig: CIFAR_50}
     \end{subfigure}
        \caption{Comparison of EPC accuracy across different tasks in the 50-vehicle scenario, with a 100 m transmission range and single-hop communication links. a) MNIST, b) GTSRB, c) CIFAR-10.}
        \label{fig: EPC Accuracy for 50 Vehicles}
\end{figure}
Figures 1 and 2 present the EPC accuracy trajectories of AERO-HMTFL and the three baselines (RHFedMTL, Ditto, and M-Fed) across communication rounds on MNIST, GTSRB, and CIFAR-10 at two vehicular densities (20 and 50 vehicles).

For MNIST, AERO-HMTFL quickly stabilizes in a high-accuracy region and sustains this plateau with limited variance under both densities. By contrast, RHFedMTL, Ditto, and M-Fed exhibit longer-lasting transient behavior and recurrent drops in accuracy. These instabilities reflect their limited ability to control update inconsistency under mobility. RHFedMTL is more affected by representation mismatch across clients, Ditto is sensitive to highly heterogeneous client updates without task-aligned coordination, and M-Fed is vulnerable to unstable clients because it lacks mechanisms to down-weight unreliable contributions. As a result, AERO-HMTFL maintains approximately 5\% higher sustained EPC accuracy than the strongest baseline in the denser scenario while remaining noticeably smoother.

For GTSRB, the separation is clearer due to the greater perceptual complexity and stronger heterogeneity. RHFedMTL exhibits higher variance before stabilizing. Ditto displays sharper oscillations as client participation changes, and M-Fed suffers from repeated collapses when inconsistent updates propagate into aggregation. In contrast, AERO-HMTFL suppresses these failure modes by jointly reducing feature inconsistency and mitigating unreliable updates through reliability-aware aggregation, resulting in approximately 2\% higher sustained accuracy than the best baseline with 20 vehicles and around 12\% higher with 50 vehicles.

On CIFAR-10, the most challenging dataset, baseline limitations become more apparent because gradient diversity is greater and the cost of unmanaged heterogeneity is higher. RHFedMTL and Ditto intermittently approach competitive accuracy but remain volatile, settling into less stable regimes, while M-Fed continues to experience pronounced drops. AERO-HMTFL remains stable and achieves the highest sustained EPC accuracy, reaching approximately 13\% improvement over the strongest baseline at the final plateau across both densities.

Overall, the results confirm that AERO-HMTFL provides superior convergence quality and predictive accuracy under short-range, single-hop vehicular communication. The baselines are primarily limited by (i) weaker robustness to representation mismatch (RHFedMTL), (ii) higher sensitivity to heterogeneous and time-varying client participation without coordination (Ditto), and (iii) susceptibility to unstable updates due to the absence of reliability-aware aggregation (M-Fed), whereas AERO-HMTFL explicitly addresses these issues through latent feature learning, reliability weighting, and hierarchical task-aware coordination.

\begin{figure}[!t]
     \centering
     \begin{subfigure}{0.45\textwidth}
         \centering
         \includegraphics[width=\textwidth]{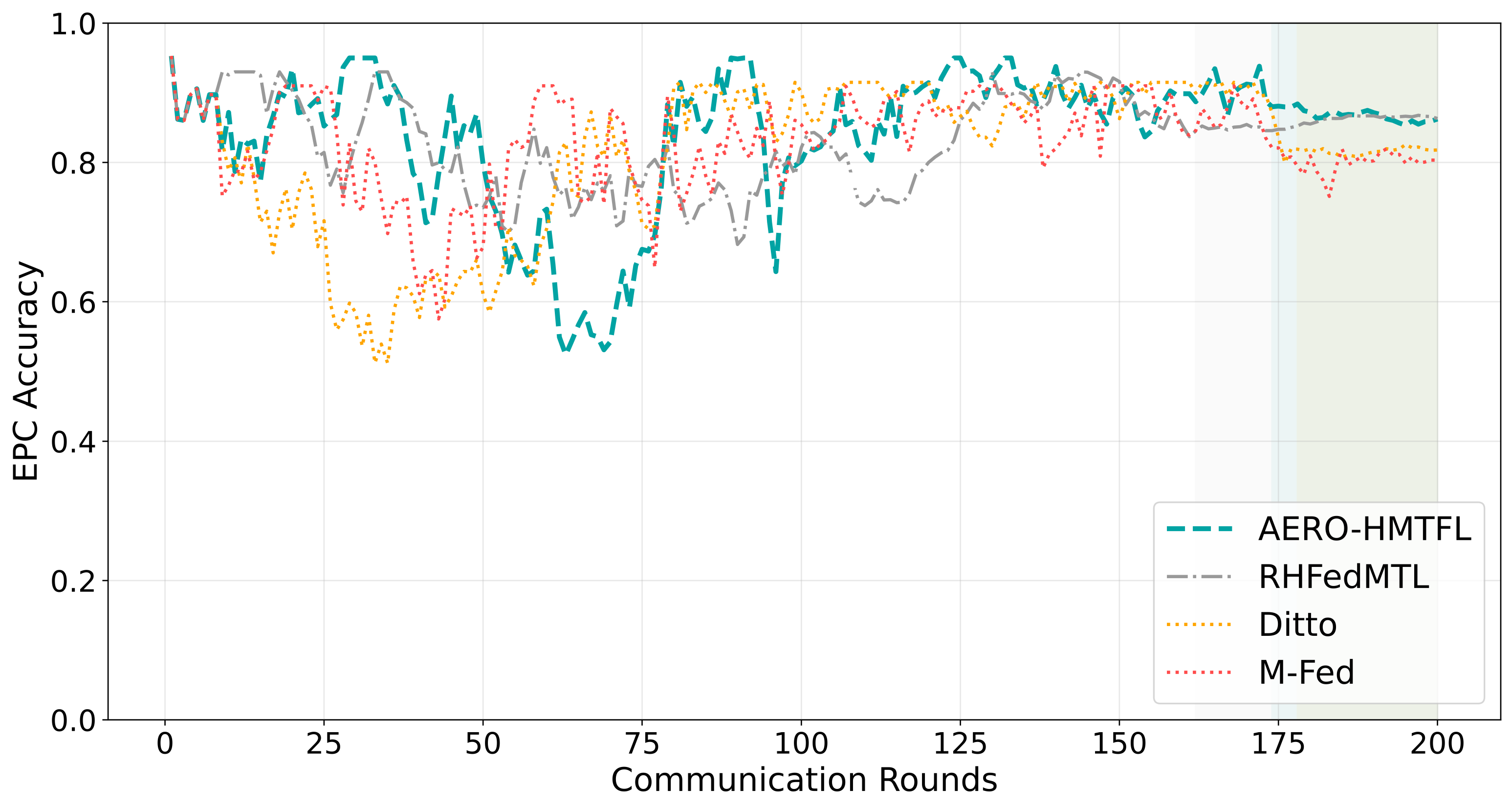}
         \subcaption*{(a)}
         \label{fig: MNIST_500_20}
     \end{subfigure}
     \vfill
     \begin{subfigure}{0.45\textwidth}
         \centering
         \includegraphics[width=\textwidth]{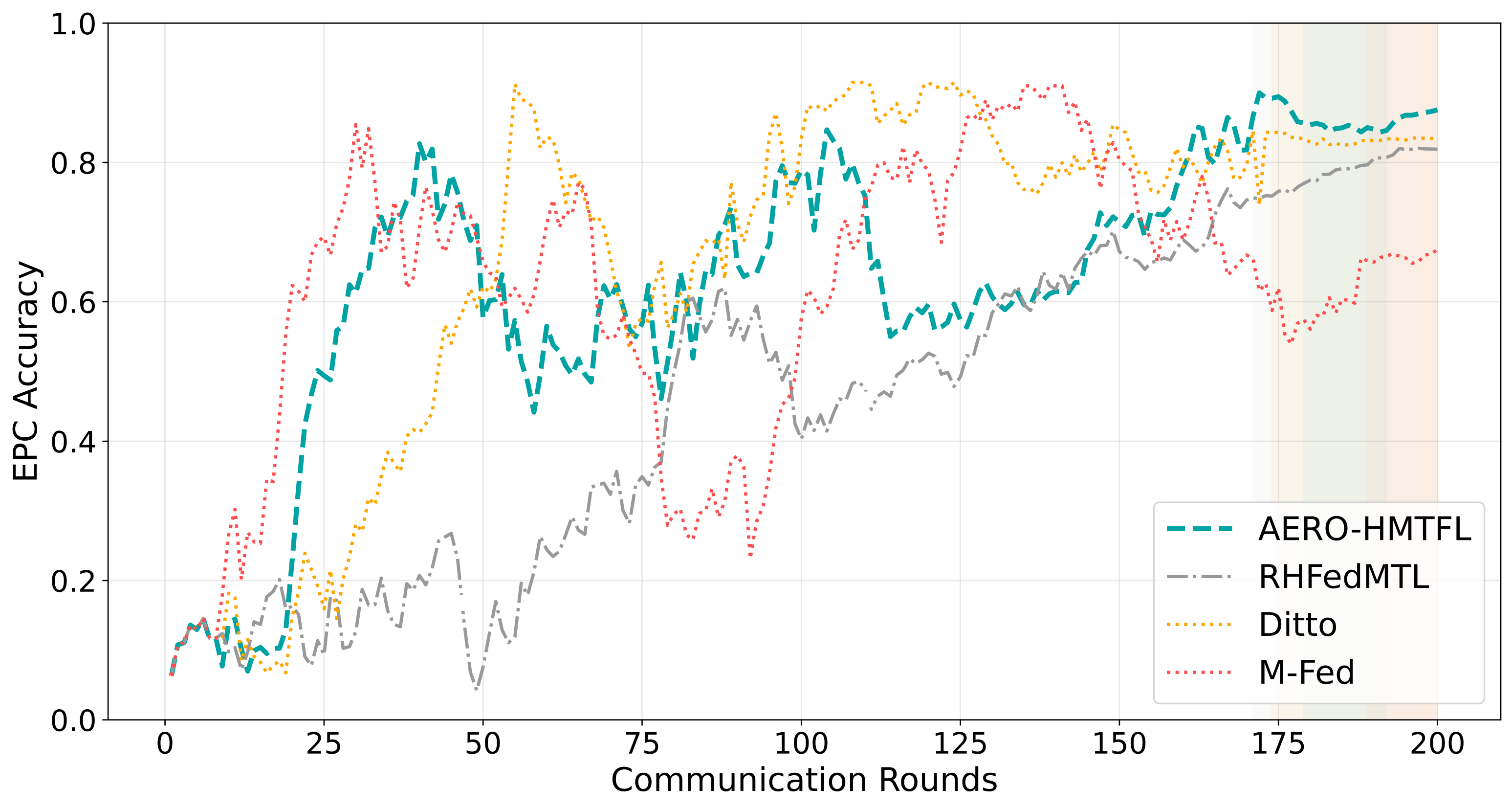}
         \subcaption*{(b)}
         \label{fig: GTSRB_500_20}
     \end{subfigure}
     \vfill
     \begin{subfigure}{0.45\textwidth}
         \centering
         \includegraphics[width=\textwidth]{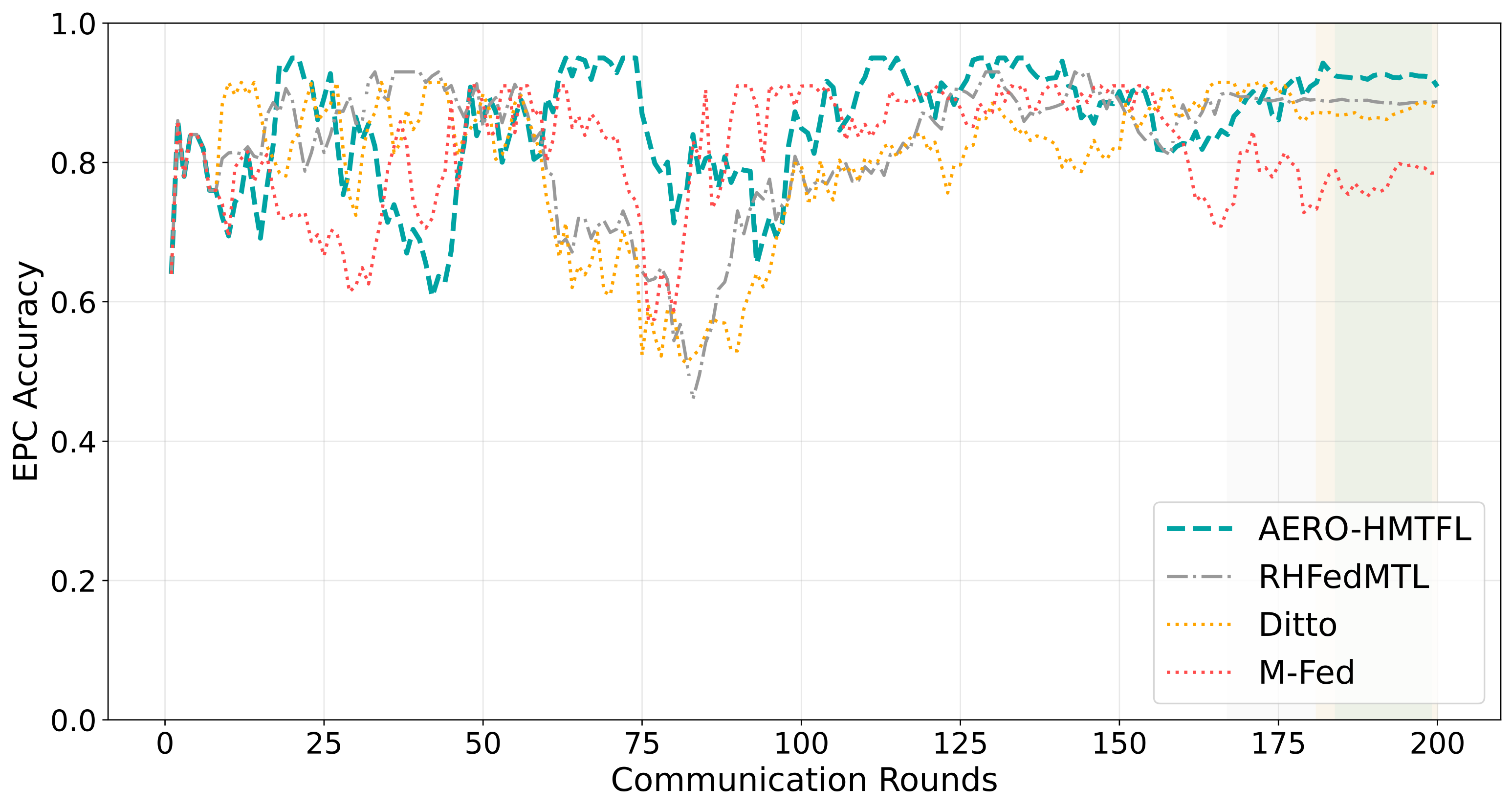}
         \subcaption*{(c)}
         \label{fig: CIFAR_500_20}
     \end{subfigure}
        \caption{Comparison of EPC accuracy across different tasks in the 20-vehicle scenario, with a 500 m transmission range and single-hop communication link. a) MNIST, b) GTSRB, c) CIFAR-10.}
        \label{fig: EPC Accuracy for 20 Vehicles}
\end{figure}

Figure 3 evaluates the EPC accuracy of AERO-HMTFL against RHFedMTL, Ditto, and M-Fed in the 20-vehicle scenario when the transmission range is increased to 500 m under single-hop communication. In shorter-range settings, the wider coverage yields a denser, less fragmented communication graph, enabling more vehicles to participate consistently in each aggregation step.

In MNIST, all methods operate in a high-accuracy regime, but AERO-HMTFL exhibits the tightest late-round variance. Ditto and M-Fed still experience intermittent drops during the transient phase, as expected, because denser connectivity also increases exposure to heterogeneous client updates, and methods without reliability control remain sensitive to inconsistent gradients. Consequently, AERO-HMTFL maintains a modest yet consistent advantage, achieving roughly 5–8\% higher sustained accuracy than the weakest baseline.

For GTSRB, the separation is more pronounced. AERO-HMTFL sustains the highest stable accuracy band, while RHFedMTL settles lower, and Ditto and M-Fed remain noticeably more oscillatory. The stronger gains on GTSRB are consistent with the task’s higher heterogeneity and perception complexity, where dense neighborhoods help only if the aggregation mechanism can suppress unreliable updates. In the final regime, AERO-HMTFL is about 8–12\% higher than RHFedMTL and can exceed 25\% over M-Fed.

On CIFAR-10, AERO-HMTFL again achieves the most stable and highest sustained EPC accuracy. Although RHFedMTL and Ditto occasionally approach similar levels, they remain more variable and typically plateau below AERO-HMTFL, while M-Fed exhibits larger fluctuations. This is consistent with CIFAR-10’s higher gradient diversity, where reliability-aware hierarchical aggregation is critical to dampen oscillations. In late rounds, AERO-HMTFL attains roughly 10–15\% higher sustained accuracy than M-Fed.

Overall, Figure 3 indicates that increasing the range to 500 m improves connectivity and reduces fragmentation. However, AERO-HMTFL benefits the most, as denser neighborhoods enhance reliability estimation and stabilize hierarchical aggregation, resulting in consistently higher EPC accuracy and smoother convergence behavior across MNIST, GTSRB, and CIFAR-10.

\begin{figure}[!t]
     \centering
     \begin{subfigure}{0.45\textwidth}
         \centering
         \includegraphics[width=\textwidth]{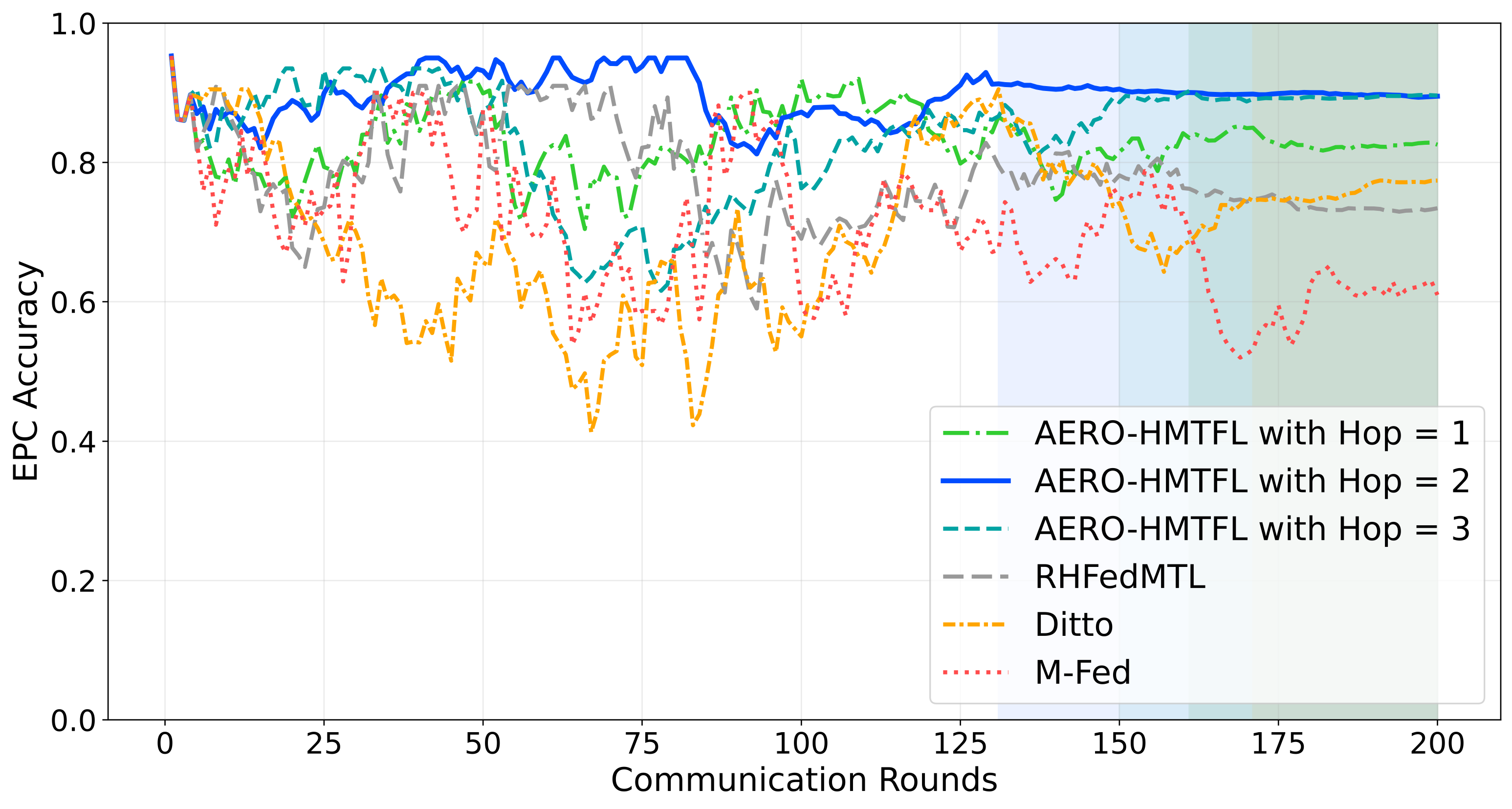}
         \subcaption*{(a)}
         \label{fig: Hops_MNIST}
     \end{subfigure}
     \vfill
     \begin{subfigure}{0.45\textwidth}
         \centering
         \includegraphics[width=\textwidth]{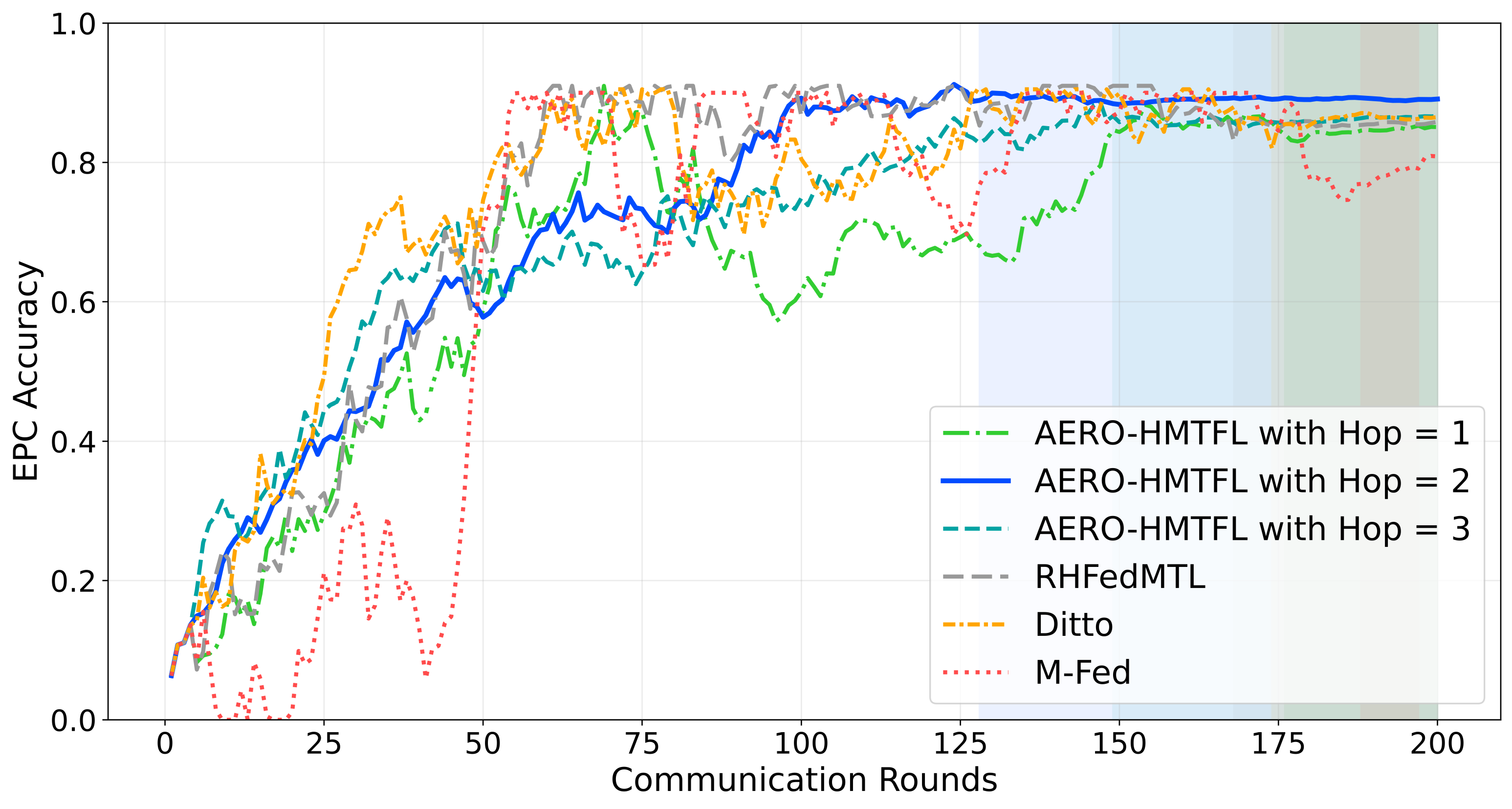}
         \subcaption*{(b)}
         \label{fig: Hops_GTSRB}
     \end{subfigure}
     \vfill
     \begin{subfigure}{0.45\textwidth}
         \centering
         \includegraphics[width=\textwidth]{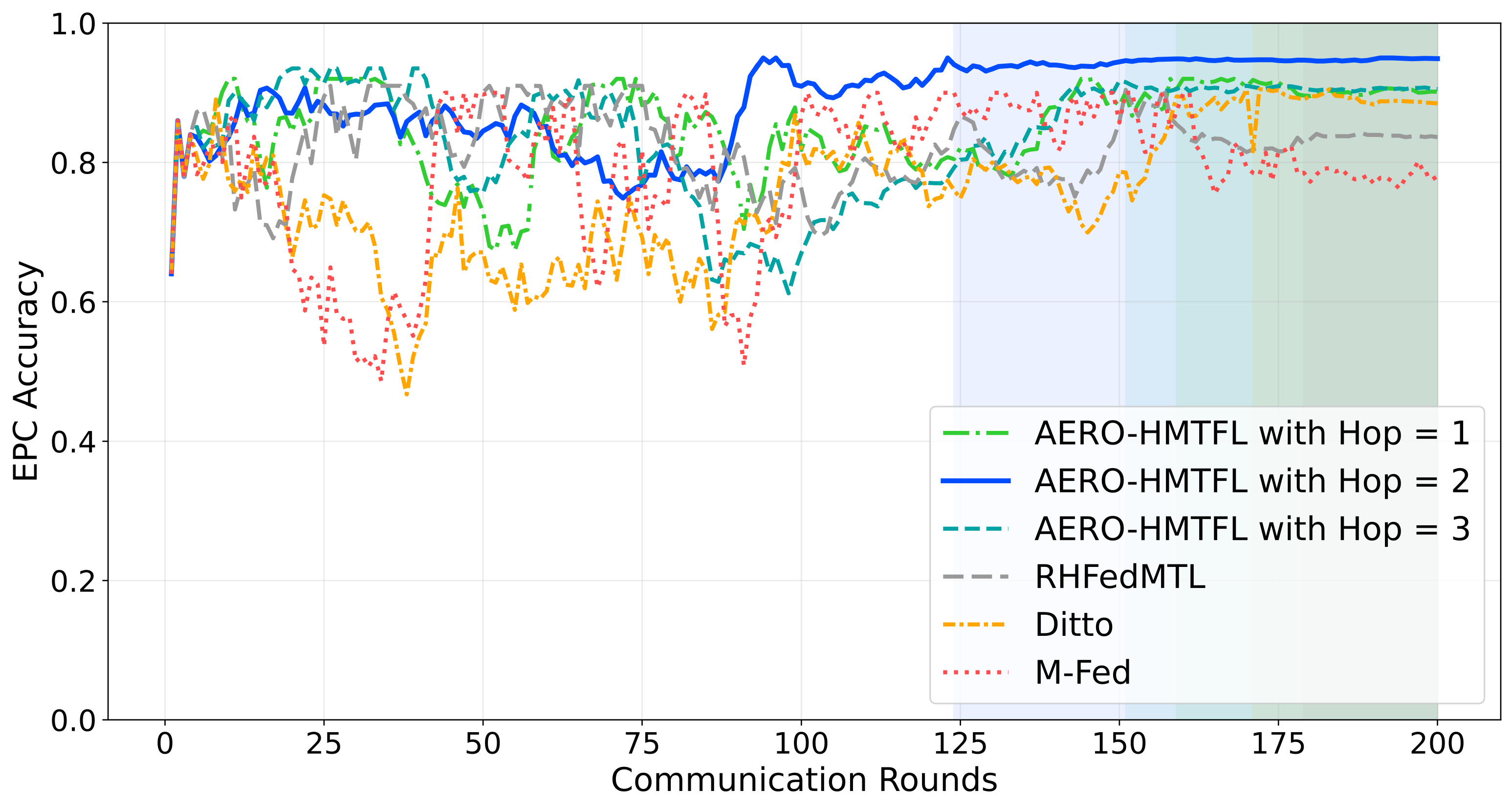}
         \subcaption*{(c)}
         \label{fig: Hops_CIFAR}
     \end{subfigure}
        \caption{Comparison of EPC accuracy across different tasks in the 20-vehicle scenario, with a 100 m transmission range for different numbers of hops. a) MNIST, b) GTSRB, c) CIFAR-10.}
        \label{fig: EPC Accuracy for hops}
\end{figure}

Figure 4 compares EPC accuracy across MNIST, GTSRB, and CIFAR-10 in the 20-vehicle scenario with a 100 m transmission range, while varying the maximum hop count for AERO-HMTFL (Hop = 1, 2, 3) and benchmarking against RHFedMTL, Ditto, and M-Fed. Increasing the hop limit expands the effective neighborhood and reduces short-range fragmentation, thereby typically improving information flow and stabilizing the aggregation process. However, overly large hop counts can also spread heterogeneous updates across a wider portion of the network, increasing gradient variance and inducing additional oscillations.

For MNIST, AERO-HMTFL with Hop = 2 achieves the highest and most stable accuracy plateau. Hop = 1 and Hop = 3 remain competitive but exhibit slightly larger fluctuations, which is consistent with the trade-off between connectivity and heterogeneity. Hop = 1 provides limited mixing for links under 100 m, while Hop = 3 broadens mixing enough to introduce stronger update inconsistencies. The baselines show deeper oscillations and lower sustained accuracy, indicating higher sensitivity to non-IID updates. In particular, RHFedMTL remains affected by residual feature mismatch; Ditto is more sensitive to heterogeneous and time-varying participation; and M-Fed exhibits greater instability due to the lack of mechanisms for hierarchical clustering and for suppressing unreliable updates.

For GTSRB, the benefit of multi-hop connectivity is more pronounced due to higher task heterogeneity. AERO-HMTFL achieves the cleanest stabilization, especially with Hop = 2, since multi-hop communication improves the representativeness of aggregation by connecting vehicles that would otherwise remain isolated at 100 m under single-hop communication. Hop = 1 shows slower, noisier improvement, consistent with a limited neighborhood size, whereas Hop = 3 can introduce additional variance by propagating mismatched updates across broader neighborhoods. The baselines remain more oscillatory, with M-Fed showing particularly unstable behavior, reflecting its vulnerability to inconsistent gradients under heterogeneous perception data, and Ditto exhibiting pronounced fluctuations due to uncontrolled mixing of diverse client updates.

On CIFAR-10, AERO-HMTFL again achieves the highest sustained EPC accuracy and the most stable late-round regime with Hop = 2. Hop = 1 tends to be constrained by reduced mixing, whereas Hop = 3 can amplify the propagation of heterogeneity, which is more costly for CIFAR-10 due to its higher gradient diversity. RHFedMTL, Ditto, and M-Fed remain below AERO-HMTFL in sustained accuracy and exhibit larger fluctuations, consistent with their weaker ability to control mobility and address update inconsistency under complex vision tasks.

Overall, Figure 4 shows that expanding connectivity via multi-hop communication can improve stability in short-range links, provided the hop count is carefully chosen. Hop = 2 consistently provides the best balance between improved neighborhood coverage and controlled heterogeneity propagation. This aligns with the design of AERO-HMTFL, where reliability-aware aggregation and hierarchical coordination benefit from richer connectivity, while excessive propagation can introduce additional variance.

\begin{table}[H]
\centering
\caption{Convergence time (in communication rounds) as a function of the number of vehicles and transmission ranges for different datasets and algorithms within the first 200 communication rounds, assuming single-hop connectivity.}
\label{tab:convergence_combined}
\begin{tabular}{|c|c|c|c|c|}
\hline
\textbf{Scenario} & \textbf{Algorithm} & \textbf{MNIST} & \textbf{GTSRB} & \textbf{CIFAR-10} \\ \hline
\multirow{4}{*}{\shortstack{\textbf{20 Vehicles}\\\textbf{Tx = 100 m}}}
                                      & \textbf{AERO-HMTFL} & \textbf{138} & \textbf{141} & \textbf{141} \\ \cline{2-5}
                                      & RHFedMTL            & 164          & 164          & 163          \\ \cline{2-5}
                                      & Ditto               & 174          & 174          & 174          \\ \cline{2-5}
                                      & M-Fed               & 194          & 184          & 184          \\ \hline
\multirow{4}{*}{\shortstack{\textbf{50 Vehicles}\\\textbf{Tx = 100 m}}}
                                      & \textbf{AERO-HMTFL} & \textbf{142} & \textbf{142} & \textbf{154} \\ \cline{2-5}
                                      & RHFedMTL            & 185          & 185          & 180          \\ \cline{2-5}
                                      & Ditto               & 189          & 195          & 188          \\ \cline{2-5}
                                      & M-Fed               & N/A          & N/A          & 195          \\ \hline
\multirow{4}{*}{\shortstack{\textbf{20 Vehicles}\\\textbf{Tx = 500 m}}}
                                      & \textbf{AERO-HMTFL} & \textbf{174} & \textbf{179} & \textbf{184} \\ \cline{2-5}
                                      & RHFedMTL            & 162          & 171          & 167          \\ \cline{2-5}
                                      & Ditto               & 178          & 174          & 181          \\ \cline{2-5}
                                      & M-Fed               & N/A          & 189          & N/A          \\ \hline
\end{tabular}
\end{table}

Table I summarizes the convergence time (in communication rounds) as a function of the number of vehicles and transmission ranges for single-hop connectivity on MNIST, GTSRB, and CIFAR-10. Across all datasets and both densities, AERO-HMTFL consistently converges faster than RHFedMTL and Ditto, and it is also more reliable than M-Fed, which fails to converge in some dense settings. With 20 vehicles, AERO-HMTFL converges approximately 13–16\% fewer rounds than RHFedMTL and 19–21\% fewer rounds than Ditto, depending on the dataset. Compared to M-Fed, AERO-HMTFL reduces the convergence time by roughly 23–29\%, indicating substantially more stable learning under short-range, mobility-affected participation. With 50 vehicles, the performance gap widens. AERO-HMTFL converges in about 14–23\% fewer rounds than RHFedMTL and 18–27\% fewer rounds than Ditto. Notably, M-Fed does not converge within the first 200 rounds for MNIST and GTSRB, highlighting its sensitivity to increased heterogeneity and mobility when learning is performed without hierarchical coordination and reliability-aware suppression of unstable updates. Overall, Table I confirms that AERO-HMTFL maintains superior scalability as the number of vehicles increases from 20 to 50 across different transmission ranges; its convergence time increases only marginally, whereas the baselines degrade more noticeably, and M-Fed may fail to converge.

Although AERO-HMTFL does not achieve the minimum convergence time at Tx = 500 m, its behavior is consistent with the accuracy trajectories observed in Fig. 3. In the 500 m regime, the communication graph becomes significantly denser and less fragmented, so the main bottleneck shifts from connectivity limitations to heterogeneity management. As a result, RHFedMTL can satisfy the convergence criterion earlier by stabilizing at a plateau under dense participation, thanks to base-station aggregation. However, Fig. 3 indicates that this earlier stabilization does not necessarily correspond to the best sustained accuracy or the most stable late-round dynamics. Dense single-hop connectivity diminishes the relative advantage of communication-aware coordination in terms of speed, allowing RHFedMTL to stabilize faster, whereas AERO-HMTFL prioritizes robustness and sustained predictive performance. This explains why AERO-HMTFL remains near the best convergence times while delivering improvements in stability and accuracy.


\begin{table}[H]
\centering
\caption{Convergence time (in communication rounds) for different numbers of hops and different datasets within the first 200 communication rounds, assuming 20 vehicles with 100 m transmission range.}
\begin{tabular}{|c|c|c|c|c|}
\hline
\textbf{Scenario} & \textbf{Algorithm} & \textbf{MNIST} & \textbf{GTSRB} & \textbf{CIFAR-10} \\ \hline
\multirow{6}{*}{\textbf{\begin{tabular}[c]{@{}c@{}}20 Vehicles\\ Tx = 100 m\end{tabular}}} & \textbf{\begin{tabular}[c]{@{}c@{}}AERO-HMTFL \\ with Hop = 1\end{tabular}} & \textbf{161}   & \textbf{176}   & \textbf{159}      \\ \cline{2-5} 
                                                                                           & \textbf{\begin{tabular}[c]{@{}c@{}}AERO-HMTFL\\ with Hop = 2\end{tabular}}  & \textbf{131}   & \textbf{128}   & \textbf{124}      \\ \cline{2-5} 
                                                                                           & \textbf{\begin{tabular}[c]{@{}c@{}}AERO-HMTFL\\ with Hop = 3\end{tabular}}  & \textbf{150}   & \textbf{149}   & \textbf{151}      \\ \cline{2-5} 
                                                                                           & RHFedMTL                                                                    & 161            & 168            & 179               \\ \cline{2-5} 
                                                                                           & Ditto                                                                       & 171            & 174            & 171               \\ \cline{2-5} 
                                                                                           & M-Fed                                                                       & N/A            & 188            & N/A               \\ \hline
\end{tabular}
\end{table}

Table~II compares convergence time within the first 200 rounds for different hop limits under Tx = 100 with 20 vehicles across MNIST, GTSRB, and CIFAR-10. The best configuration is consistently AERO-HMTFL with Hop = 2, which converges fastest on all datasets. Relative to Hop = 1, Hop = 2 reduces convergence time by 18.6\% (MNIST), 27.3\% (GTSRB), and 22.0\% (CIFAR-10). It also outperforms Hop = 3 by 12.7\%, 14.1\%, and 17.9\%, respectively. Compared with RHFedMTL and Ditto, Hop = 2 improves convergence by roughly 19--31\% and 23--28\%, respectively, while M-Fed fails to converge for MNIST and CIFAR-10 within 200 rounds.

\begin{figure}[ht]
\setlength\belowcaptionskip{0pt}
\centering
\includegraphics[width=9 cm]{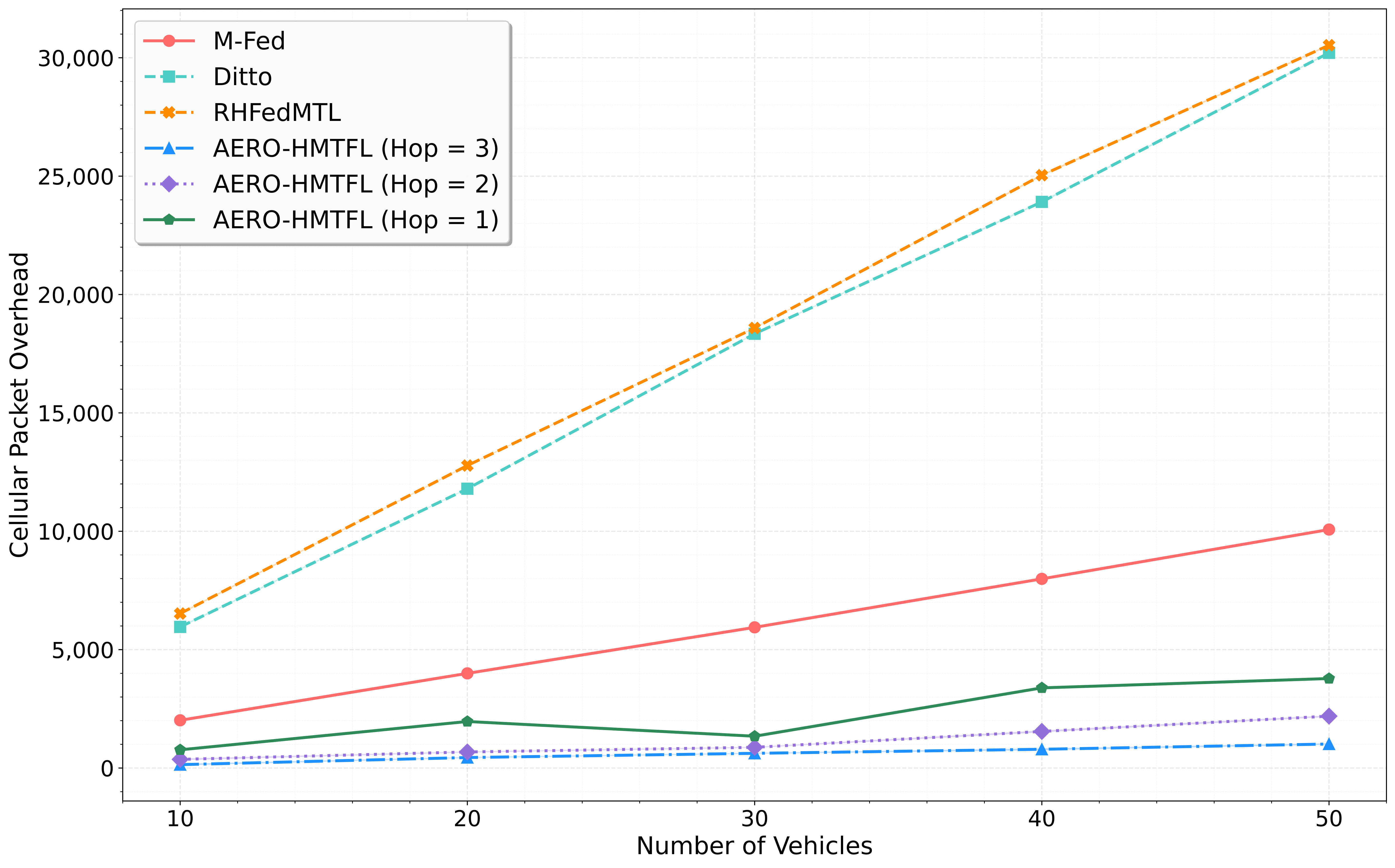}
\caption{Cellular packet overhead as a function of the number of vehicles}
\end{figure}

Figure~5 shows cellular packet overhead measured as the total
number of packets sent to the EPC as a function of the number of vehicles at a transmission range of 100~m. Two key observations arise. First, Ditto and RHFedMTL incur the highest overhead and scale almost linearly with the network size. Both methods rely on frequent client-level transmissions without effective mobility-aware communication localization. Second, M-Fed exhibits lower overhead than Ditto and RHFedMTL, as expected, since autoencoder-based latent representations can reduce the update footprint. However, M-Fed still requires substantially more packets than AERO-HMTFL as the system scales because it lacks clustering and hierarchical coordination to localize aggregation. Across all vehicle densities, AERO-HMTFL achieves the lowest packet count and the best scalability. At 50 vehicles, AERO-HMTFL reduces packet transmissions by approximately 87-97\% relative to Ditto/RHFedMTL, depending on the hop setting (about $3.8\times 10^{3}$, $2.2\times 10^{3}$, and $1.0\times 10^{3}$ packets for Hop~=~1, 2, and 3, respectively). This reduction directly follows from addressing the baseline limitations: hierarchical, task-aware coordination localizes aggregation and limits redundant network-wide exchanges.

\section{Conclusion}

This paper proposed AERO-HMTFL, an AutoEncoder-based Reliability-Optimized Hierarchical Multi-Task Federated Learning framework for dynamic multi-hop clustered VANETs. Unlike conventional vehicular HFL methods that assume a single common task, AERO-HMTFL enables vehicles with heterogeneous learning objectives to collaborate through a shared autoencoder-based representation while retaining task-specific heads locally. The framework integrates task-aware multi-hop clustering, split-model representation sharing, reliability-aware aggregation, and hierarchical coordination to support scalable and privacy-preserving multi-task learning under mobility, intermittent connectivity, and time-varying participation. By incorporating task affinity into the clustering criterion, the proposed method promotes both mobility stability and semantic alignment among collaborating vehicles. Moreover, exchanging only shared autoencoder parameters, task identifiers, and aggregate validation scores preserves raw-data and task-model privacy, while reliability-based weighting limits the influence of unstable or low-quality updates. Simulation results demonstrate that AERO-HMTFL outperforms RHFedMTL, Ditto, and M-Fed in convergence speed, learning stability, accuracy, and communication efficiency. Across the evaluated scenarios, it reduces convergence time by up to approximately 29\% and communication overhead by about 87-97\% in dense settings. The multi-hop analysis further indicates that controlled connectivity is essential, with the two-hop configuration providing the most effective balance between information exchange and heterogeneity propagation. Overall, AERO-HMTFL provides an integrated learning architecture that jointly addresses task heterogeneity, mobility, communication constraints, and update reliability in VANETs. Future work will consider more complex perception tasks, multimodal sensor data, adaptive task scheduling, and security-aware aggregation under adversarial conditions.

\section*{Disclosure of AI-Assisted Writing}

The authors used Open AI ChatGPT to assist with grammar correction, English-language editing, and improving the clarity and readability of the manuscript. All AI-assisted content was reviewed and approved by the authors, who take full responsibility for the final manuscript.

\bibliographystyle{ieeetr}
\bibliography{bare_jrnl_clean}

\end{document}